\documentclass[onefignum,onetabnum]{siamart220329}

\usepackage{lipsum}
\usepackage{amsfonts}
\usepackage{graphicx}
\usepackage{epstopdf}
\usepackage{algorithmic}
\ifpdf
  \DeclareGraphicsExtensions{.eps,.pdf,.png,.jpg}
\else
  \DeclareGraphicsExtensions{.eps}
\fi

\newsiamremark{remark}{Remark}
\newsiamremark{hypothesis}{Hypothesis}
\crefname{hypothesis}{Hypothesis}{Hypotheses}
\newsiamthm{claim}{Claim}

\headers{Shortest Paths in Graphs of Convex Sets}{T. Marcucci, J. Umenberger, P. A. Parrilo, and R. Tedrake}

\title{
Shortest Paths in Graphs of Convex Sets\thanks{
Submitted to the editors DATE.
\funding{
National Science Foundation, award no. EFMA-1830901.
Department of the Navy, Office of Naval Research, award no. N00014-18-1-2210 and award no. N00014-17-1-2699.
}
}
}

\author{
Tobia Marcucci\thanks{
Department of Electrical Engineering and Computer Science, Massachusetts Institute of Technology, Cambridge, MA
 (\email{tobiam@mit.edu}, \email{umnbrgr@mit.edu}, \email{parrilo@mit.edu}, \email{russt@mit.edu}).}
\and Jack Umenberger\footnotemark[2]
\and Pablo A. Parrilo\footnotemark[2]
\and Russ Tedrake\footnotemark[2]
}

\usepackage{amsopn}

\usepackage[many]{tcolorbox}
\usepackage{bm}
\usepackage[caption=false]{subfig}

\newsiamremark{example}{Example}
\newsiamremark{assumption}{Assumption}

\newcommand{\R}{\mathbb R}

\newcommand{\cB}{\mathcal B}

\newcommand{\cD}{\mathcal D}

\newcommand{\cI}{\mathcal I}
\newcommand{\cJ}{\mathcal J}
\newcommand{\cL}{\mathcal L}
\newcommand{\cN}{\mathcal N}
\newcommand{\cS}{\mathcal S}

\newcommand{\cX}{\mathcal X}
\newcommand{\cY}{\mathcal Y}
\newcommand{\cV}{\mathcal V}
\newcommand{\cE}{\mathcal E}
\newcommand{\cP}{\mathcal P}
\newcommand{\cK}{\mathcal K}

\newcommand{\bzero}{\bm 0}
\newcommand{\bone}{\bm 1}
\newcommand{\ba}{\bm a}
\newcommand{\bb}{\bm b}
\newcommand{\bc}{\bm c}
\newcommand{\bd}{\bm d}

\newcommand{\bp}{\bm p}
\newcommand{\bq}{\bm q}
\newcommand{\br}{\bm r}
\newcommand{\bs}{\bm s}
\newcommand{\bv}{\bm v}
\newcommand{\bx}{\bm x}
\newcommand{\by}{\bm y}
\newcommand{\bz}{\bm z}
\newcommand{\bA}{\bm A}
\newcommand{\bB}{\bm B}
\newcommand{\bC}{\bm C}
\newcommand{\bR}{\bm R}
\newcommand{\bZ}{\bm Z}
\newcommand{\btheta}{\bm \theta}

\DeclareMathOperator{\conv}{conv}
\DeclareMathOperator{\epi}{epi}

\DeclareMathOperator{\cl}{cl}
\DeclareMathOperator{\ext}{ext}

\newcommand{\inc}{^\mathrm{in}}
\newcommand{\out}{^\mathrm{out}}

\newcommand{\tand}{\mathrm{\ and \ }}

\newcommand{\tforall}{\mathrm{\ for \ all \ }}
\newcommand{\minimize}{\mathrm{minimize}}
\newcommand{\maximize}{\mathrm{maximize}}
\newcommand{\subjectto}{\mathrm{subject \ to}}

\ifpdf
\hypersetup{
pdftitle={Shortest Paths in Graphs of Convex Sets},
pdfauthor={T. Maruccci, J. Umenberger, P. A. Parrilo, and R. Tedrake}
}
\fi

\usepackage{xcolor}

\begin{document}

\maketitle

\begin{abstract}
Given a graph, the shortest-path problem requires finding a sequence of edges with minimum cumulative length that connects a source vertex to a target vertex.
We consider a variant of this classical problem in which the position of each vertex in the graph is a continuous decision variable constrained in a convex set, and the length of an edge is a convex function of the position of its endpoints.
Problems of this form arise naturally in
many areas, from motion planning of autonomous vehicles to optimal control of hybrid systems.
The price for such a wide applicability is the complexity of this problem, which is easily seen to be NP-hard.
Our main contribution is a strong and lightweight mixed-integer convex formulation based on perspective operators, that makes it possible to efficiently find globally optimal paths in large graphs and in high-dimensional spaces.
\end{abstract}

\begin{keywords}
Shortest-path problem, graph problems with neighborhoods, mixed-integer convex programming, perspective formulation, optimal control.
\end{keywords}

\begin{MSCcodes}
52B05, 90C11, 90C25, 90C35, 90C57, 93C55, 93C83.
\end{MSCcodes}

\section{Introduction}
\label{sec:intro}
\begin{figure}[t]
\centering
\includegraphics[height=3.2cm]{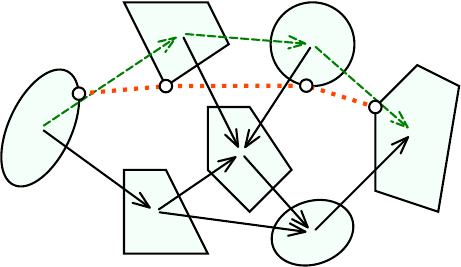}
\caption{
Example of an SPP in GCS.
The source set is on the left and the target set is on the right.
The graph edges are arrows, and the shortest path is shown in dashed green.
The dotted red lines connect the optimal positions of the vertices along the shortest path.
}
\label{fig:toy_example}
\end{figure}

The Shortest-Path Problem (SPP) is one of the most important and ubiquitous problems in combinatorial optimization.
In its single-source single-target version, this problem asks for a path of minimum length connecting two prescribed vertices of a graph, where the length of a path is defined as the sum of the lengths of its edges.
Typically, the edge lengths are fixed scalars, given as problem data, and the assumptions made on their values have a dramatic impact on the problem complexity~\cite[Chapters~6 to~8]{schrijver2003combinatorial}.
In this paper we introduce the SPP in Graph of Convex Sets (GCS), a variant of the SPP in which the edge lengths are convex functions of continuous variables representing the position of the vertices (see Figure~\ref{fig:toy_example}).
More precisely, a GCS is a directed graph in which each vertex is paired with a convex set.
The spatial position of a vertex is a continuous variable, constrained to lie in the corresponding convex set.
The length of an edge is a given convex function of the position of the vertices that this edge connects.
When looking for a path of minimum length in a GCS, we then have the extra degree of freedom of optimizing the position of the vertices visited by the path.
According to the literature, this problem could also be classified as an SPP \emph{with neighborhoods}; we call it SPP in GCS to highlight the crucial role that convexity plays in the developments of this paper.

Many problems of practical interest can be formulated as SPPs in GCS: for some of those the convex sets and the edge-length functions are naturally suggested by the application, for others the construction of the GCS requires more thinking.
As an example of the former class of problems, scheduling the flight of a drone with limited batteries is immediately cast as an SPP in GCS like the one in Figure~\ref{fig:toy_example}.
The start region is on the left, the goal region is on the right, and the remaining regions can be used for recharging.
Pairs of regions that are close enough for the drone to fly between are connected by an edge.
The objective is to minimize the overall length of the flight.
Optimal control of discrete-time hybrid dynamical systems~\cite{bemporad1999control} is a main application that we target in this paper, and is an example of a problem whose formulation as an SPP in GCS is nontrivial.
In this case we let the convex sets live in the joint state and control space of the dynamical system.
Each discrete time step corresponds to an edge transition in the GCS, and the edge lengths quantify, e.g., the energy consumed to move between states (a length that is infinite if the motion is not compatible with the system dynamics).
This is explained in detail in Section~\ref{sec:optimal_control}.

\subsection{Contributions}

The following are the main contributions of this article.

\subsubsection*{Problem statement (Section~\ref{sec:statement})} 
The SPP in GCS represents an unexplored class of problems at the interface of combinatorial and convex optimization.
It lends itself to a simple problem statement and, at the same time, it is a versatile framework that includes as special cases many problems of practical relevance.

\subsubsection*{Mixed-integer convex formulation (Section~\ref{sec:micp})} 
The SPP in GCS is easily seen to be NP-hard (Section~\ref{sec:complexity}).
Our main contribution is the formulation of this problem as a strong and lightweight Mixed-Integer Convex Program (MICP).
This program extends in a natural way the classical network-flow formulation of the SPP,
and it allows us to efficiently find shortest paths in large graphs (hundreds of vertices) and high-dimensional spaces (tens of dimensions).
In addition, the design principles of this MICP can be applied to improve existing mixed-integer formulations of other graph problems with neighborhoods, which are limited to small graphs and sets in two or three dimensions (see Appendix~\ref{sec:extensions}).

\subsubsection*{Set-based convex relaxation of bilinear constraints (Section~\ref{sec:relaxation})} 
The main building block of our MICP is a tight and compact convex relaxation for a class of bilinear constraints that emerge naturally in our problem.
This relaxation is \emph{set based}, in the sense that it does not rely on the explicit constraints that define the sets in our GCS, but it works directly with their abstract set representations.
This makes our MICP usable even when these sets are black boxes accessible only through a separation oracle.
This relaxation is similar in spirit to the Lov\'asz-Schrijver one~\cite{lovasz1991cones}, and is based on perspective operators (a popular tool in mixed-integer optimization~\cite{ceria1999convex,stubbs1999branch,frangioni2006perspective,gunluk2010perspective}).

\subsubsection*{Control applications (Section~\ref{sec:optimal_control})} 
Computation times are the main limitation to a widespread application of mixed-integer optimization in control of hybrid systems~\cite{naik2017embedded,stellato2018embedded,marcucci2020warm}.
Our shortest-path formulation of these problems is substantially different from the state of the art~\cite{moehle2015perspective,marcucci2019mixed}, as we do not use binary variables to encode the discrete mode in which the system is at each time step but, instead, we use them to select the transitions between modes.
This different parameterization yields slightly larger but much stronger MICPs that, in our computational experiments, are orders of magnitude faster to solve.

\subsection{Related graph problems}
\label{sec:related_works}
In this subsection we overview a few variants of classical graph problems that are closely related to our problem formulation.

\subsubsection*{Graph problems with neighborhoods}
\label{sec:graph_problems_with_neighborhoods}

Graph problems where the vertices are allowed to move within corresponding sets are often called problems with neighborhoods.
The SPP with neighborhoods has been analyzed in~\cite{disser2014rectilinear} under stringent assumptions that ensure polynomial-time solvability: the sets are disjoint rectilinear polygons in the plane, and the edge lengths penalize the $\cL_1$ distance between the vertices.
The applications we target with this paper, however, do not verify any of these hypotheses.
A special case of the SPP with neighborhoods is the touring-polygon problem, which asks for the shortest path between two points that visits a set of polygons in a given order~\cite{dror2003touring}.
Our problem differs from this in that our sets are convex and the order in which we visit them is not predefined.
Other problems akin to the touring polygon, but substantially different from the SPP in GCS, are the safari, the zookeeper, and the watchman route; see~\cite[Part~IV]{li2011euclidean} and the references therein.

The Traveling-Salesman Problem (TSP) and the Minimum-Spanning-Tree Problem (MSTP) are the two combinatorial problems that have been studied most extensively in their variants with neighborhoods~\cite{arkin1994approximation,yang2007minimum}.
Exact algorithms for these generally rely on expensive mixed-integer nonconvex optimization~\cite{gentilini2013travelling,blanco2017minimum,burdick2021multi}, and do not scale beyond two or three dimensions.
Although the techniques we propose in this paper are particularly well suited to the structure of the SPP, they can be used without modifications to formulate other graph problems with neighborhoods as very tractable MICPs (see Appendix~\ref{sec:extensions}).

\subsubsection*{Graph problems with clusters}

Generalized Steiner problems~\cite{dror2000generalizedsteiner} (otherwise known as generalized network-design problems~\cite{feremans2003generalized,pop2012generalized}) can be thought as the discrete counterpart of the graph problems with neighborhoods: the vertex set is partitioned into clusters and the problem constraints are expressed in terms of these clusters, rather than the original vertices.
A clustered version of the SPP has been presented in~\cite{li1995shortest}: each vertex in the graph is assigned a nonnegative weight, and the total vertex weight incurred by the shortest path within each cluster must not exceed a given value.
The problem we analyze in this paper can be approximated as an SPP with clusters in a natural way.
In low-dimensional spaces, this approximation can be computationally efficient and sufficiently accurate for practical applications.
However, this strategy is infeasible in high dimensions, where covering a volume of space with a cluster requires an exponential number of points.

\subsubsection*{Euclidean shortest paths}

Another related variant of the SPP is the Euclidean SPP~\cite{li2011euclidean}, where we look for a continuous path that connects two points and avoids given polygonal obstacles.
In two dimensions, this problem can be reduced to a discrete search and is solvable in polynomial time~\cite{lozano1979algorithm}.
In three dimensions or more the problem is NP-hard~\cite[Theorem~2.3.2]{canny1987new}, and common algorithms rely on a grid discretization of the space~\cite{kim2003discrete}.
More recently, a moment-based technique that handles semialgebraic obstacles has been proposed in~\cite{khadir2020piecewise}.

\section{Problem statement}
\label{sec:statement}
We start with a formal statement of the SPP in GCS.
Let $G := (\cV, \cE)$ be a directed graph with vertex set $\cV$ and edge set $\cE$.
For each vertex $v \in \cV$, we have a nonempty compact convex set $\cX_v \subset \R^n$ and a point $\bx_v$ contained in it.\footnote{The results presented in this paper are easily extended to the case in which the sets $\cX_v$ do not have common dimension $n$.}
The length of an edge $e=(u,v) \in \cE$ is determined by the location of the points $\bx_u$ and $\bx_v$ via the expression $\ell_e(\bx_u,\bx_v)$.
The \emph{edge length} function $\ell_e$ takes values in $\R_{\geq 0} \cup \{\infty\}$ and is assumed to be proper, closed, and convex.
Note that, despite its name, we do not assume $\ell_e$ to be a valid metric, and properties like symmetry or the triangle inequality are not required to hold.
Given a source vertex $s$ and a target vertex $t  \neq s$, an $s$-$t$ path $p$ is a sequence of distinct vertices $(v_0, \ldots, v_K)$ such that $v_0=s$, $v_K=t$,  and $(v_k, v_{k+1})\in \cE$ for all $k=0, \ldots, K-1$.
We denote with $\cE_p := \{(v_0, v_1), \ldots, (v_{K-1}, v_K)\}$ the set of edges traversed by this path, and with $\cP$ the set of all $s$-$t$ paths in the graph $G$.
The SPP in GCS is then stated as
\begin{subequations}
\label{eq:spp_in_gcs}
\begin{tcolorbox}[ams align]
\label{eq:spp_objective}
\minimize \quad & \sum_{e = (u,v) \in \cE_p} \ell_e(\bx_u,\bx_v) \\
\subjectto \quad
\label{eq:spp_path}
& p \in \mathcal P, \\
\label{eq:spp_v}
& \bx_v \in \cX_v, && \forall v \in p.
\end{tcolorbox}\noindent
\end{subequations}
The decision variables are the discrete path $p$ and the continuous values $\bx_v$.
The cost~\eqref{eq:spp_objective} minimizes the total path length.
Constraint~\eqref{eq:spp_v} is enforced only for the vertices visited by the path, since the positions of the other vertices are irrelevant.

The edge length used in Figure~\ref{fig:toy_example} is the Euclidean distance:
\begin{align}
\label{eq:2norm}
\ell_e (\bx_u, \bx_v) := \|\bx_v - \bx_u\|_2.
\end{align}
With this choice the polygonal line connecting the points $\bx_v$ along a shortest path is as straight as possible, perfectly straight if $(s,t) \in \cE$.
Conversely, if the edge length is the Euclidean distance squared,
\begin{align}
\label{eq:2norm_squared}
\ell_e (\bx_u, \bx_v) := \|\bx_v - \bx_u\|_2^2,
\end{align}
straight trajectories may be suboptimal if they require long steps  $\bx_v - \bx_u$.
Note also that by letting $\ell_e$ take infinite value outside a convex set $\cX_e$ we are effectively enforcing the edge constraint $(\bx_u,\bx_v) \in \cX_e$.
This will be used in Section~\ref{sec:optimal_control} to formulate optimal-control problems as SPPs in GCS: there the edge constraints will couple the vertex positions according to the system dynamics.

\section{Complexity analysis}
\label{sec:complexity}
If we fix the vertex positions $\bx_v$, problem~\eqref{eq:spp_in_gcs} simplifies to the classical SPP with scalar nonnegative edge lengths, which is easily solvable using, e.g., Linear Programming (LP).
Similarly, if we fix the path $p$, problem~\eqref{eq:spp_in_gcs} simplifies to a convex program that can be efficiently solved for most convex sets $\cX_v$ and edge lengths $\ell_e$.
In this section we show that the simultaneous optimization of the vertex positions and the path makes the SPP in GCS an NP-hard problem.

Recall that an $s$-$t$ path $p := (v_0, \ldots, v_K)$ is said to be Hamiltonian if it visits every vertex in the graph (i.e., if $K = |\cV|-1$), and a graph is Hamiltonian if it contains such a path.
The Hamiltonian-Path Problem (HPP) asks if a given graph is Hamiltonian.
As an example, the graph in Figure~\ref{fig:toy_example} is not Hamiltonian.

\begin{theorem}
\label{th:complexity}
The SPP in GCS~\eqref{eq:spp_in_gcs} is NP-hard.
\end{theorem}

\begin{proof}
We show that the HPP is polynomial-time reducible to the SPP in GCS.
The thesis will then follow since the HPP is NP-complete~\cite{karp1972reducibility}.
We construct an SPP in GCS that shares the same graph $G$ as the given HPP.
We let the source $\cX_s := \{0\}$ and target $\cX_t := \{1\}$ sets be singletons on the real line, while we define $\cX_v := [0,1]$ for all $v \neq s, t$.
The length of each edge is the Euclidean distance squared~\eqref{eq:2norm_squared}.
Given these choices, the optimal positioning of the vertices for a fixed path $p$ is given by $\bx_{v_k} = k /K$ for $k=0, \ldots, K$.
The length of this path is $K (1 / K)^2 = 1 / K$.
We conclude that an optimal path is one for which $K$ is maximized, and is Hamiltonian if and only if $G$ is Hamiltonian.
This reduction operates in polynomial time.
\end{proof}

This simple reduction shows that, even if the convex sets $\cX_v$ are one-dimensional intervals, the SPP in GCS can be a hard problem.
Nonetheless, one might wonder if additional assumptions on the problem data could turn the SPP in GCS into a problem that is solvable in polynomial time.
\begin{itemize}
\item
What if the graph $G$ is acyclic?
In case of an acyclic graph the HPP is solvable in linear time~\cite[Section~4.4]{ahuja1993network}, and our hardness proof is not valid.
\item
What if the sets $\cX_v$ are disjoint?
In fact, some graph problems with neighborhoods can be solved more efficiently  in case of disjoint neighborhoods~\cite{disser2014rectilinear,burdick2021multi}.
\item
What if the edge lengths $\ell_e$ are positively homogeneous?
An edge length like the Euclidean distance~\eqref{eq:2norm} could not be used in our reduction since it would not force the optimal path $p$ to visit as many vertices as possible.
\end{itemize}

It turns out that all these questions have a negative answer.
This is summarized in the following theorem, whose proof is omitted since it is a long and relatively straightforward adaptation of the complexity analysis of the Euclidean SPP from~\cite{canny1987new}.

\begin{theorem}
\label{th:complexity2}
Assume that the graph $G$ is acyclic, the sets $\cX_v$ are disjoint, and the edge lengths $\ell_e$ are positively homogeneous.
The SPP in GCS~\eqref{eq:spp_in_gcs} is NP-hard.
\end{theorem}

\section{Convex-analysis background}
\label{sec:tools}
This section introduces two basic concepts in convex analysis: perspective operators (homogenization) and valid inequalities (duality).
These are the main tools that we will use in the design and the analysis of our MICP.
Our goal here is to set the notation and collect some important definitions and properties; for a comprehensive introduction to these topics see~\cite[Parts~II and~III]{rockafellar1970convex} or~\cite[Chapters~III and~IV]{hiriart2013convex}.

\subsection{Perspective operators}
\label{sec:perspective}

There is a natural construction that maps a set in $n$ dimensions to a cone in $n+1$ dimensions.
This is sometimes called \emph{homogenization}, or the \emph{cone over} the set.
Here we call it \emph{perspective}, for coherence with the name commonly used for the same operation applied to functions~\cite[Section~IV.2.2]{hiriart2013convex}.

\begin{definition}
\label{def:perspective_cone}
We define the \emph{perspective} of  a closed convex set $\cX \subseteq \R^n$ as
$$
\tilde \cX := \cl \{(\bx, \lambda) : \lambda \geq 0, \ \bx \in \lambda \cX \},
$$
where $\cl$ denotes the closure of the set.
\end{definition}

\begin{remark}
\label{rem:bounded_set}
The closure operation  in Definition~\ref{def:perspective_cone} is unnecessary for bounded sets $\cX$.
While, when the set  $\cX$ is unbounded, it ensures that the perspective $\tilde \cX$ contains all its limit points with $\lambda = 0$~\cite[Theorem~8.2]{rockafellar1970convex}.
\end{remark}

Importantly, the perspective operation preserves convexity, and the set $\tilde \cX$ is a closed convex cone.
The next example shows that the perspective of a set represented in conic form can be computed very easily.

\begin{example}
\label{ex:perspective_conic}
Let $\cX := \{\bx : \bA \bx + \bb \in \cK \}$, for some matrix $\bA$, vector $\bb$, and closed convex cone $\cK$.
We have $\tilde \cX = \{(\bx, \lambda) : \lambda \geq 0, \bA \bx + \bb \lambda \in \cK \}$.
\end{example}

This example has great practical relevance since, informally, it tells us that if a conic-optimization solver can handle the set $\cX$ then it can also handle its perspective $\tilde \cX$.
For instance, we see that the perspective of a polyhedral, ellipsoidal, and spectrahedral set can be represented through a set of linear, second-order-cone, and semidefinite constraints, respectively.
More in general, if the set $\cX$ is bounded, the formal equivalence of optimizing over $\cX$ and its perspective $\tilde \cX$ can be established using the ellipsoid method~\cite[Chapter~4]{grotschel2012geometric}, since the separation problems for these two sets are easily seen to be equivalent.

The next definition uses the construction from~\cite[Page~39]{rockafellar1970convex} to describe what the perspective operation does to a convex function.

\begin{definition}
\label{def:perspective_function}
We define the \emph{perspective} of a closed convex function $f: \R^n \rightarrow \R \cup \{\infty\}$ as the unique function $\tilde f$ whose epigraph is the perspective of the epigraph of $f$, i.e.,
$$
\tilde f (\bx, \lambda) := \inf \{\sigma : (\bx, \sigma, \lambda) \in \widetilde{\epi f}\},
$$
where $\epi f := \{(\bx, \sigma) : f(\bx) \leq \sigma\}$.\footnote{More precisely the sets $\epi \tilde f$ and $\widetilde{\epi f}$ are isomorphic:
$\epi \tilde f := \{(\bx, \lambda, \sigma): (\bx, \sigma, \lambda) \in \widetilde{\epi f}\}$.
}
\end{definition}

Since its epigraph is closed and convex, the perspective function $\tilde f$ is closed and jointly convex in $\bx$ and $\lambda$.

\begin{remark}
\label{rem:perspective_values}
For $\lambda > 0$, noticing that $\lambda \epi f = \{(\bx, \sigma) : \lambda f(\bx / \lambda) \leq \sigma\}$, we have that the perspective function is $\tilde f(\bx, \lambda) = \lambda f(\bx / \lambda)$.
For $\lambda < 0$, we immediately see that $\tilde f(\bx, \lambda) = \infty$.
The behavior for $\lambda = 0$ is more complicated~\cite[Corollary~8.5.2]{rockafellar1970convex}, but for the scope of this paper it suffices to note that if $f$ is proper then, by the closedness of $\tilde f$, we must have $\tilde f(\bzero, 0) = 0$.
\end{remark}

Although Definition~\ref{def:perspective_function} might seem unsuitable for numerical optimization, the perspectives of most common functions $f$ can be minimized using standard solvers.
In fact, given a conic representation of the epigraph of $f$, we can compute the epigraph of $\tilde f$ as in Example~\ref{ex:perspective_conic}, and minimize $\tilde f$ using a slack variable.

The next two examples draw further useful parallels between the perspective operation applied to sets and to functions.

\begin{example}
\label{ex:perspective_extended_values}
Let $\cX$ be a nonempty closed convex set and $g$ be a finite convex function.
Define $f(\bx) := g(\bx)$ if $\bx \in \cX$ and $f(\bx) := \infty$ otherwise.
We have $\tilde f (\bx, \lambda) = \tilde g (\bx, \lambda)$ if $(\bx, \lambda) \in \tilde \cX$ and $\tilde f (\bx, \lambda) = \infty$ otherwise.
\end{example}

\begin{example}
\label{ex:perspective_functional_description}
For a set $\cX := \{\bx : f_i(\bx) \leq 0 \tforall i \in \cI\}$, where the functions $f_i$ are closed and convex, we have $\tilde \cX = \{(\bx, \lambda) : \tilde f_i (\bx, \lambda) \leq 0 \tforall i \in \cI \}$.
Equivalently, using Remark~\ref{rem:perspective_values}, we have $\tilde \cX = \cl \{(\bx, \lambda) : \lambda > 0, \lambda f_i (\bx / \lambda) \leq 0 \tforall i \in \cI \}$.
\end{example}

\subsection{Valid inequalities}

A second cone that is naturally associated with a convex set is the cone of its valid inequalities.
This will play an important role in the analysis of our MICP in Section~\ref{sec:relaxation}.
We report here a formal definition and a useful property.

\begin{definition}
\label{def:valid_inequalities}
We define the \emph{cone of valid inequalities} of a set  $\cX \subseteq \R^n$ as
$$
\cX^\circ := \{(\ba, b) : \ba^\top \bx + b \geq 0 \tforall \bx \in \cX\}.
$$
\end{definition}

The cone $\cX^\circ$ is easily seen to be closed and convex, even when $\cX$ is neither closed nor convex.
Note also that the cone of valid inequalities is closely related to the \emph{polar set}, but the latter lives in $n$ dimensions.

The next lemma relates the two operations defined in this section.
Recall that the \emph{dual cone} of a closed convex cone $\cK$ is the set $\cK^* := \{\ba : \ba^\top \bx \geq 0 \tforall \bx \in \cK\}$.

\begin{lemma}
\label{lemma:polar_description}
Let $\cX$ be a closed convex set.
The closed convex cones $\tilde \cX$ and $\cX^\circ$ are dual to each other.
\end{lemma}

\begin{proof}
The perspective cone $\tilde \cX$ can be equivalently defined as the closure of the cone generated by $\cX \times \{1\}$.
By applying~\cite[Corollary~11.7.2]{rockafellar1970convex} to the latter set, we obtain $\tilde \cX = \{ (\bx, \lambda) : \ba^\top \bx + b \lambda \geq 0 \tforall (\ba, b) \in \cX^\circ \}$.
This shows that $\tilde \cX = (\cX^\circ)^*$.
The other direction follows from the bipolar theorem~\cite[Theorem~14.1]{rockafellar1970convex}.
\end{proof}

\section{Mixed-integer convex formulation}
\label{sec:micp}
We now present the main contribution of this paper: the formulation of the SPP in GCS~\eqref{eq:spp_in_gcs} as a strong and lightweight MICP.
This program is designed in two steps.
First, in Section~\ref{sec:bilinear}, we extend the network-flow formulation of the classical SPP (recalled in Section~\ref{sec:network_flow}) to our setting.
This yields an optimization problem with bilinear equality constraints.
Second, in Section~\ref{sec:convexification_bilinearities}, we construct a convex relaxation tailored to these bilinear constraints and we formulate our MICP.
The relaxation technique used in this section will be described at a higher level of generality and thoroughly analyzed in Section~\ref{sec:relaxation}.

\subsection{Network-flow formulation of the SPP}
\label{sec:network_flow}
The starting point for the design of our MICP is the network-flow formulation of the SPP with scalar nonnegative edge lengths (see, e.g.,~\cite[Section~4.1]{ahuja1993network}):
\begin{subequations}
\label{eq:network_flow}
\begin{align}
\label{eq:network_flow_objective}
\minimize
\quad & \sum_{e \in \cE} l_e y_e \\
\subjectto \quad
\label{eq:network_st}
& \sum_{e \in \cE_s\out} y_e =1, \ \sum_{e \in \cE_t\inc} y_e = 1, \\
\label{eq:network_v}
& \sum_{e \in \cE_v\inc} y_e  = \sum_{e \in \cE_v\out} y_e, \ \sum_{e \in \cE_v\out} y_e \leq 1, && \forall v \in \cV - \{s,t\}, \\
\label{eq:network_flow_nonnegativity}
& y_e \geq 0, && \forall e \in \cE.
\end{align}
\end{subequations}
In this LP the decision  variables $y_e$ parameterize a path $p$, with $y_e = 1$ if the edge $e$ is traversed by $p$ and $y_e = 0$ otherwise.
The scalar $l_e \geq 0$ represents the length of the edge $e$.
The sets $\cE_v\inc := \{(u,v) \in \cE \}$ and $\cE_v\out := \{(v,u) \in \cE \}$ collect the edges incoming and outgoing vertex $v$.
Without loss of generality, we assume $|\cE_s\inc| = |\cE_t\out| = 0$, i.e., the source and the target have no incoming and outgoing edges, respectively.
Interpreting the value of $y_e$ as the \emph{flow} carried by the edge $e$, constraint~\eqref{eq:network_st} asks that one unit of flow is injected in the source and ejected from the target.
For all the other vertices, constraint~\eqref{eq:network_v} enforces the \emph{flow conservation} and a \emph{degree constraint}.
The latter enforces a limit of one to the total flow traversing the vertex.

\begin{remark}
\label{rem:integrality}
Note that we do not explicitly require the flows $y_e$ to be binary, but we only enforce their nonnegativity in~\eqref{eq:network_flow_nonnegativity}.
This is because all the basic feasible solutions of the LP~\eqref{eq:network_flow} can be shown to have binary value (see, e.g.,~\cite[Section~11.12]{ahuja1993network}), and the constraints $y_e \in \{0,1\}$ would not affect the optimal value of this program.
\end{remark}

\begin{remark}
\label{rem:degree}
Since we assumed the edge lengths $l_e$ to be nonnegative, the degree constraint in~\eqref{eq:network_v} is actually redundant for the LP~\eqref{eq:network_flow}, as well as for problem~\eqref{eq:bilinear_spp} below.
However, as we will see in Section~\ref{sec:degree_constraints}, this constraint is not redundant for our final MICP.
Therefore we include it in our formulation from the start.
\end{remark}

\subsection{Biconvex formulation}
\label{sec:bilinear}
As an intermediate step towards our MICP, we formulate the SPP in GCS as a \emph{biconvex} optimization problem.
Specifically, a nonlinear program whose nonconvexity comes only from products between the vertex locations and the flow variables parameterizing a path.
Note that this is consistent with the observation from Section~\ref{sec:complexity} that the SPP in GCS simplifies to a convex program if we fix either the vertex locations or the path through the graph.

A natural attempt to extend the LP~\eqref{eq:network_flow} to the SPP in GCS is to proceed as done for other graph problems with neighborhoods~\cite{gentilini2013travelling,blanco2017minimum,burdick2021multi}: include the vertex locations $\bx_v$ among our decision variables, enforce the constraint $\bx_v \in \cX_v$ for all $v \in \cV$, and substitute the addends in the cost~\eqref{eq:network_flow_objective} with $\ell_e(\bx_u, \bx_v) y_e$.
However, one immediate issue with this approach is that the latter product is undefined if $\ell_e(\bx_u, \bx_v) = \infty$ and $y_e = 0$, while we would like the cost contribution of the edge $e$ to always be zero if $y_e = 0$.
Perspective functions give us a convenient and rigorous way to ``turn on and off'' the length of an edge using the corresponding flow variable.

Let us introduce two auxiliary variables $\bz_e := y_e \bx_u$ and $\bz_e' := y_e \bx_v$ per edge $e = (u,v)$, and consider the perspective function $\tilde \ell_e(\bz_e, \bz_e', y_e)$.\footnote{
We are slightly abusing notation here: since in Definition~\ref{def:perspective_function} we defined the perspective of functions with a single argument, to be precise, we should write $\tilde \ell_e((\bz_e, \bz_e'), y_e)$.
}
When the flow $y_e$ is positive, this function coincides with the product above:
$$
\tilde \ell_e(\bz_e, \bz_e', y_e)
= \ell_e (\bz_e/y_e, \bz_e'/y_e) y_e
= \ell_e (y_e \bx_u/y_e, y_e \bx_v/y_e) y_e
= \ell_e(\bx_u, \bx_v) y_e,
$$
where the first equality comes from Remark~\ref{rem:perspective_values}.
When the flow $y_e$ is zero, the function $\tilde \ell_e$ is well defined and correctly evaluates to zero, even when $\ell_e(\bx_u, \bx_v) = \infty$.
In fact, $y_e = 0$ implies $\bz_e = \bz_e' = \bzero$, and $\tilde \ell_e(\bzero,\bzero,0) = 0$ as discussed in Remark~\ref{rem:perspective_values}.

Overall, we then have the following biconvex formulation of the SPP in GCS:
\begin{subequations}
\label{eq:bilinear_spp}
\begin{tcolorbox}[ams align]
\label{eq:bilinear_objective}
\minimize
\quad & \sum_{e \in \cE} \tilde \ell_e(\bz_e, \bz_e', y_e) \\
\subjectto \quad
\label{eq:bilinear_flow}
& \text{constraints of problem~\eqref{eq:network_flow}}, \\
\label{eq:bilinear_Xv}
& \bx_v \in \cX_v, && \forall v \in \cV, \\
\label{eq:bilinear_yz}
& \bz_e = y_e \bx_u, \ \bz_e'= y_e \bx_v, && \forall e = (u,v) \in \cE.
\end{tcolorbox}\noindent
\end{subequations}
The decision variables are the flows $y_e$, the vertex positions $\bx_v$, and the auxiliary variables $\bz_e$ and $\bz_e'$.
The role of the latter is to match the vertices $\bx_u$ and $\bx_v$ when $y_e = 1$, and collapse to zero when $y_e = 0$.
This behavior is driven by the bilinear equality constraints~\eqref{eq:bilinear_yz}, which are the only nonconvexity in our formulation and whose convexification is the focus of the next subsection.
Before that, let us formally verify that, as mentioned in Remark~\ref{rem:integrality} for the LP~\eqref{eq:network_flow}, forcing the flows $y_e$ to be binary does not affect the optimal value of the biconvex program~\eqref{eq:bilinear_spp}.

\begin{proposition}
\label{prop:integrality}
For any local minimum $L \in \R_{\geq 0}$ of problem~\eqref{eq:bilinear_spp}, there exists a feasible point of~\eqref{eq:bilinear_spp} with cost equal to $L$ and such that $y_e \in \{0,1\}$ for all $e \in \cE$.
\end{proposition}

\begin{proof}
Given a local minimizer of~\eqref{eq:bilinear_spp} with cost $L$, we fix the vertex positions $\bx_v$.
This reduces problem~\eqref{eq:bilinear_spp} to an LP of the form~\eqref{eq:network_flow}.
The optimal value of this LP must be $L$, otherwise we would have found a descent direction and our solution of~\eqref{eq:bilinear_spp} would not be locally optimal.
Furthermore, because of Remark~\ref{rem:integrality}, we can assume that the optimal flows of this LP are binary.
Paired with the previously fixed variables $\bx_v$, these binary flows yield a feasible solution of~\eqref{eq:bilinear_spp} with cost $L$.
\end{proof}

\subsection{Convex relaxation of the bilinear constraints}
\label{sec:convexification_bilinearities}
The biconvex program~\eqref{eq:bilinear_spp} is our first formulation of the SPP in GCS that can be tackled numerically.
However, the bilinear constraints~\eqref{eq:bilinear_yz} make this optimization problem challenging to solve, even just locally.
In this subsection we show how to reformulate problem~\eqref{eq:bilinear_spp} as a lightweight and strong MICP, that can be reliably solved to global optimality using branch-and-bound algorithms.

The next lemma allows us to construct a tight envelope around the constraints of the biconvex program~\eqref{eq:bilinear_spp} through a small number of perspective cones.
In its statement we let $\cE_v := \cE_v\inc \cup \cE_v\out$ denote the set of edges incident with vertex $v \in \cV$.
Recall also that a valid constraint for an optimization problem is a constraint that is verified by all the feasible points.
\begin{lemma}
\label{lemma:valid_constraint}
For some vertex $v \in \cV$, assume that the linear inequality
\begin{align}
\label{eq:valid_inequality}
\sum_{e \in \cE_v} c_e y_e + d \geq 0
\end{align}
is valid for problem~\eqref{eq:bilinear_spp}.
Partitioning the summation over $\cE_v$ in incoming and outgoing edges, we have that the following convex constraint is also valid for~\eqref{eq:bilinear_spp}:
\begin{align}
\label{eq:implied_constraint}
\left(\sum_{e \in \cE_v\inc} c_e \bz_e' + \sum_{e \in \cE_v\out} c_e \bz_e + d \bx_v, \ \sum_{e \in \cE_v} c_e y_e + d \right)
\in \tilde \cX_v.
\end{align}
\end{lemma}

\begin{proof}
Constraint~\eqref{eq:implied_constraint} requires two conditions to hold.
One is~\eqref{eq:valid_inequality}, which is assumed.
The other is verified by multiplying both sides of $\bx_v \in \cX_v$ from~\eqref{eq:bilinear_Xv} by the left-hand side of~\eqref{eq:valid_inequality}, and then using the bilinear constraints~\eqref{eq:bilinear_yz}.
\end{proof}

\begin{remark}
\label{rem:valid_equality}
If the valid constraint~\eqref{eq:valid_inequality} holds with equality, Lemma~\ref{lemma:valid_constraint} simply amounts to multiplying this equality by $\bx_v$, and it gives us a valid linear equality of the form $\sum_{e \in \cE_v\inc} c_e \bz_e' + \sum_{e \in \cE_v\out} c_e \bz_e + d \bx_v = \bzero$.
\end{remark}

\begin{remark}
\label{rem:rlt}
Generating new valid constraints by multiplying existing ones is a standard procedure at the core of many relaxation techniques~\cite{mccormick1976computability,sherali1990hierarchy,lovasz1991cones,lasserre2001global,parrilo2003semidefinite}.
Lemma~\ref{lemma:valid_constraint} will be analyzed at a higher level of generality in Section~\ref{sec:relaxation}, where its similarities with existing methods will be clearly drawn.
\end{remark}

Lemma~\ref{lemma:valid_constraint} lifts any valid linear constraint on the flows incident with vertex~$v$ into a convex constraint that envelops the feasible set of problem~\eqref{eq:bilinear_spp}.
Our MICP is obtained by applying this lemma to each flow constraint in the LP~\eqref{eq:network_flow}, and by replacing the constraints of the biconvex program~\eqref{eq:bilinear_spp} with the envelope resulting from this process.
Let us first state our MICP and then prove its equivalence to the SPP in GCS (Theorem~\ref{th:validity_micp} below):
\begin{subequations}
\label{eq:micp}
\begin{tcolorbox}[ams align]
\label{eq:micp_objective}
\minimize
\quad & \sum_{e \in \cE} \tilde \ell_e(\bz_e, \bz_e', y_e) \\
\subjectto \quad
\label{eq:micp_st}
& \sum_{e \in \cE_s\out} y_e =1, \ \sum_{e \in \cE_t\inc} y_e = 1, \\
\label{eq:micp_degree}
& \sum_{e \in \cE_v\out} y_e \leq 1, && \forall v \in \cV - \{s,t\}, \\
\label{eq:micp_conservation}
& \sum_{e \in \cE_v\inc} (\bz_e', y_e) = \sum_{e \in \cE_v\out} (\bz_e, y_e), && \forall v \in \cV - \{s,t\}, \\
\label{eq:micp_nonnegativity}
& (\bz_e, y_e) \in \tilde \cX_u, \  (\bz_e', y_e)  \in \tilde \cX_v, && \forall e = (u,v) \in \cE,\\
\label{eq:micp_integrality}
& y_e \in \{0,1\}, && \forall e \in \cE.
\end{tcolorbox}
\end{subequations}
\noindent
Constraint~\eqref{eq:micp_conservation} is obtained as in Remark~\ref{rem:valid_equality} from the flow conservation in~\eqref{eq:network_v}, and~\eqref{eq:micp_nonnegativity} is the result of applying Lemma~\ref{lemma:valid_constraint} to the nonnegativity constraint~\eqref{eq:network_flow_nonnegativity}.
Note that the application of the same technique to the equalities~\eqref{eq:micp_st} and to the degree constraint in~\eqref{eq:network_v} would give us
\begin{subequations}
\label{eq:reconstruct}
\begin{align}
\label{eq:reconstruct_st}
& \bx_s = \sum_{e \in \cE_s\out} \bz_e, \
\bx_t = \sum_{e \in \cE_t\inc} \bz_e', \\
\label{eq:reconstruct_v}
& \bx_v - \sum_{e \in \cE_v\out} \bz_e \in \left(1 -\sum_{e \in \cE_v\out} y_e \right) \cX_v, && \forall v \in \cV - \{s,t\}.
\end{align}
\end{subequations}
However these constraints would be redundant since the vertex positions $\bx_v$ for $v \in \cV$ do not appear in the rest of problem~\eqref{eq:micp}.
Note also that the combination of~\eqref{eq:reconstruct} and~\eqref{eq:micp_nonnegativity} implies the constraints $\bx_v \in \cX_v$ for all $v \in \cV$, which would also be redundant for our MICP.
The convex relaxation of~\eqref{eq:micp} is obtained simply by dropping the integrality constraint~\eqref{eq:micp_integrality} (the nonnegativity of the flows $y_e$ is imposed by the cost and also by~\eqref{eq:micp_nonnegativity}).
Observe that, unlike the biconvex program~\eqref{eq:bilinear_spp}, the optimal value of the MICP can decrease if the flows are allowed to be fractional.

\begin{theorem}
\label{th:validity_micp}
The MICP~\eqref{eq:micp} has optimal value equal to the SPP in GCS~\eqref{eq:spp_in_gcs}.
An optimal path $p$ for problem~\eqref{eq:spp_in_gcs} is recovered from the solution of~\eqref{eq:micp} through the relation $\cE_p := \{e \in \cE: y_e = 1\}$.
An optimal positioning of the vertices is reconstructed for the source and the target as in~\eqref{eq:reconstruct_st}, and for all the other vertices by letting $\bx_v$ be any point such that~\eqref{eq:reconstruct_v} holds.
\end{theorem}

In Section~\ref{sec:conic} we will see  that this theorem follows from a simple geometric property of Lemma~\ref{lemma:valid_constraint}.
Here we give a direct proof that explicitly illustrates the logic behind our formulation.

\begin{proof}
The only flows that can satisfy the constraints in~\eqref{eq:micp} are such that the set $\cE_p$ defined above describes the vertex-disjoint union of an $s$-$t$ path and cycles.
However, the presence of cycles can be excluded since the edge lengths $\ell_e$ are nonnegative, and traversing a cycle that is disjoint from the main path cannot decrease the cost.
Therefore, at optimality, $\cE_p$ identifies a path $p$.
For all the edges $e \notin \cE_p$,  constraint~\eqref{eq:micp_nonnegativity} simplifies to $\bz_e = \bz_e' = \bzero$, and the corresponding cost addends give $\tilde \ell_e(\bzero, \bzero, 0) = 0$.
For the vertices $v \notin p$, constraint~\eqref{eq:micp_conservation} is trivially satisfied and~\eqref{eq:reconstruct_v} reads $\bx_v \in \cX_v$.
For an edge $e =(u,v)$ along the path $p$, constraint~\eqref{eq:micp_nonnegativity} becomes $\bz_e \in \cX_u$ and $\bz_e' \in \cX_v$, and the cost addend is $\tilde \ell_e(\bz_e, \bz_e', 1) = \ell_e(\bz_e, \bz_e')$.
Denoting with $f=(v,w) \in \cE_p$ the edge after $e$ in the path, the flow conservation~\eqref{eq:micp_conservation} reads $\bz_e' = \bz_f$.
Finally, the conditions in~\eqref{eq:reconstruct} give us $\bx_u = \bz_e$ and $\bx_v = \bz_e'$ for all edges $e =(u,v) \in \cE_p$.
\end{proof}

\begin{remark}
\label{rem:singletons}
If the sets $\cX_v$ are singletons, the SPP in GCS simplifies to the SPP with nonnegative edge lengths.
In this case our MICP reduces to the LP~\eqref{eq:network_flow} and its convex relaxation is exact, as discussed in Remark~\ref{rem:integrality}.
\end{remark}

\begin{remark}
\label{rem:set_based}
For the edge lengths $\ell_e$ and convex sets $\cX_v$ that typically appear in practice, the MICP~\eqref{eq:micp} can be solved to global optimality with standard solvers (see the discussion in Section~\ref{sec:perspective}).
However, problem~\eqref{eq:micp} can be tackled numerically even when the sets in our GCS are not defined by explicit constraints (e.g., convex inequalities).
For example, each set $\cX_v$ may be very complex and accessible only through an oracle that, given a point $\bx_v$, either certifies that $\bx_v \in \cX_v$ or returns a separating hyperplane.
In fact, such an oracle is easily adapted to checking membership to the perspective cones in~\eqref{eq:micp_nonnegativity}, and this black-box access to the problem constraints is sufficient for efficient optimization algorithms like the ellipsoid method~\cite{grotschel2012geometric}.
\end{remark}

\subsection{Degree constraints}
\label{sec:degree_constraints}
In Remark~\ref{rem:degree} we anticipated that, although redundant for the LP~\eqref{eq:network_flow}, the degree constraints~\eqref{eq:micp_degree} play an important role in our MICP.
This is illustrated in the next example, which shows how the optimal flows from~\eqref{eq:micp} can induce cycles if the degree constraints are not enforced.

\begin{example}
Consider a graph with vertices $\cV:= \{s,1,2,t\}$ and edges $\cE:= \{(s,1),(1,2),(2,1),(1,t)\}$.
Define the sets $\cX_s := \{-1\}$, $\cX_1 := [-1,1]$, $\cX_2 := \{0\}$, and $\cX_t := \{1\}$.
Let the length of each edge be the Euclidean distance squared~\eqref{eq:2norm_squared}.
The optimal value of this SPP in GCS is $2$ and the optimal path is $p=(s, 1, t)$.
However, if we do not enforce the degree constraints~\eqref{eq:micp_degree}, our MICP has optimal value equal to $1$ and its optimal flows are $y_e =1$ for all $e \in \cE$, i.e., they induce the cycle $(1,2,1)$.
\end{example}

In case of an acyclic graph $G$, the issues just described do not arise and the degree constraints are redundant for our MICP and its convex relaxation.
Nonetheless, we still include them in our formulation since they are computationally light and their explicit presence can trigger the use of specialized generalized-upper-bound branching rules in the solver~\cite[Section~9.2]{conforti2014integer}.


\section{Alternative formulations}
\label{sec:alternative_formulations}
Multiple alternative MICP formulations of the SPP in GCS can be designed and finding the most effective one is a tradeoff between the size of the program and the tightness of its convex relaxation.
The MICP~\eqref{eq:micp} is very compact: it has only $O(|\cE|)$ binary variables, $O(n|\cE|)$ continuous variables, and $O(n(|\cV| + |\cE|))$ constraints (assuming that the cones in~\eqref{eq:micp_nonnegativity} are described by $O(n)$ constraints).
In addition, in the experiments in Section~\ref{sec:examples}, we will see that the relaxation of our MICP is typically very tight (although a carefully designed instance in Section~\ref{sec:example_failures} shows that our relaxation can, in principle, be arbitrarily loose).

A simple alternative way to reformulate the biconvex problem~\eqref{eq:bilinear_spp} as an MICP is to enforce the integrality constraints $y_e \in \{0,1\}$ for all $e \in \cE$, and relax each bilinear constraint~\eqref{eq:bilinear_yz} independently using a McCormick envelope~\cite{mccormick1976computability}.
With our notation, this amounts to replacing~\eqref{eq:bilinear_yz} with
\begin{align}
\label{eq:mc_relaxation}
(\bz_e, y_e) \in \tilde \cB_u, \
(\bz_e', y_e) \in \tilde \cB_v, \
(\bx_u - \bz_e, 1- y_e) \in \tilde \cB_u, \
(\bx_v - \bz_e', 1- y_e) \in \tilde \cB_v,
\end{align}
where, for each $v \in \cV$, we let $\cB_v$ be an axis-aligned box that contains $\cX_v$.
Especially if the convex sets $\cX_v$ are defined by many constraints, this MICP is more compact than ours.
However, as we will see in Section~\ref{sec:examples}, this formulation has loose convex relaxation and its solution times are generally much larger than with our approach.

At the other end of the spectrum, a variety of stronger but potentially more expensive formulations could be devised.
For example, we have found that subtour-elimination constraints like~\cite[Section~2.2]{taccari2016integer} can tighten the relaxation of our MICP for some classes of problems.
Alternatively, we could formulate our MICP by grouping the constraints in~\eqref{eq:bilinear_spp} vertex by vertex, and by computing the convex hull of each group (see Section~\ref{sec:tightness_relaxation} below).
We could also use more expensive semidefinite relaxations~\cite{lovasz1991cones,lasserre2001global,parrilo2003semidefinite} of the bilinear constraints~\eqref{eq:bilinear_yz}.
In our computational experience, the MICP~\eqref{eq:micp} represents the best compromise between a lightweight and a strong formulation, and its solution times are lower than any other formulation we have tested.

\section{Analysis of the mixed-integer formulation}
\label{sec:relaxation}
In this section we describe and analyze at a more abstract level the method used in Section~\ref{sec:convexification_bilinearities} to formulate the SPP in GSC as an MICP.
We show that Lemma~\ref{lemma:valid_constraint} can be used to design convex relaxations of a large class of bilinear constraints, and we connect this result to existing relaxation techniques for nonconvex optimization.
Finally, we give a simpler geometric proof of the validity of our MICP (already shown in Theorem~\ref{th:validity_micp}).

\subsection{Set-based relaxation of bilinear constraints}
\label{sec:set_based_relaxation}

Our first step in this analysis is to show that Lemma~\ref{lemma:valid_constraint} is, in fact, a general-purpose relaxation technique for nonconvex sets of the form
\begin{align}
\label{eq:bilinear_set}
\cS := \{(\bx, \by, \bZ) : \bx \in \cX, \ \by \in \cY, \ \bZ = \bx \by^\top\},
\end{align}
where $\cX \subseteq \R^n$ and $\cY \subseteq \R^m$ are closed convex sets.
In particular, here $\cX$ takes the place of a generic set $\cX_v$ in our GCS, while $\cY$ plays the role of the linear constraints on the flow variables incident with vertex $v$ (see Remark~\ref{rem:coincise_bilinear} below for more details).

A natural approach to construct a convex envelope around the set $\cS$ is to multiply all the valid inequalities $\ba^\top \bx + b \geq 0$ for the set $\cX$ by all the valid inequalities $\bc^\top \by + d \geq 0$ for the set $\cY$, and then use the bilinear equality $\bZ = \bx \by^\top$ to linearize these products.
This gives us an infinite family of valid linear inequalities for $\cS$, which form our convex relaxation:
\begin{multline}
\label{eq:relaxation}
\cS \subseteq \cS' := \{(\bx, \by, \bZ) : \ba^\top \bZ \bc + d \ba^\top \bx + b \bc^\top \by + b d \geq 0 \\
\tforall (\ba, b) \in \cX^\circ \tand (\bc, d) \in \cY^\circ\}.
\end{multline}
Note that the conditions $\bx \in \cX$ and $\by \in \cY$ are implied by the inequalities in~\eqref{eq:relaxation} that correspond to $(\bzero, 1) \in \cY^\circ$ and $(\bzero, 1) \in \cX^\circ$, respectively.

The relaxation~\eqref{eq:relaxation} is not obviously implementable on a computer, since it involves an infinite number of constraints.
However, if one of the two sets is a polytope (i.e., a bounded polyhedron) then the convex set $\cS'$ can be efficiently described by a finite number of perspective-cone constraints.

\begin{proposition}
\label{prop:relaxation_polytopic}
Let $\cY$ be a polytope with halfspace representation $\{ \by: \bc_i^\top \by + d_i \geq 0 \tforall i \in \cI \}$.
We have
\begin{align}
\label{eq:relaxation_polytopic}
\cS' =
\{ (\bx, \by, \bZ) :
(\bZ \bc_i  + d_i\bx, \ \bc_i^\top \by + d_i)
\in \tilde \cX \tforall i \in \cI
\}.
\end{align}
\end{proposition}

\begin{proof}
To recover~\eqref{eq:relaxation} from~\eqref{eq:relaxation_polytopic} we first use Lemma~\ref{lemma:polar_description} to rewrite the membership to $\tilde \cX$ as $\ba^\top (\bZ \bc_i  + d_i \bx) + b (\bc_i^\top \by + d_i) \geq 0$ for all $(\ba, b) \in \cX^\circ$.
Then we notice that listing only the valid inequalities $(\bc_i, d_i)$ for $i \in \cI$ is equivalent to listing all the valid inequalities $(\bc, d) \in \cY^\circ$.
In fact, since $\cY$ is bounded, any vector $(\bc, d) \in \cY^\circ$ can be expressed as $\sum_{i \in \cI} \alpha_i (\bc_i, d_i)$ for some nonnegative coefficients $\alpha_i$.
Using these coefficients, the inequality $\ba^\top \bZ \bc + d \ba^\top \bx + b \bc^\top \by + b d \geq 0$ is seen to be implied by the inequalities generated by $(\bc_i, d_i)$ for $i \in \cI$.
\end{proof}

We then have two descriptions of the relaxation $\cS'$: the symmetric one~\eqref{eq:relaxation} that clearly exposes the logic behind the technique, and the asymmetric one~\eqref{eq:relaxation_polytopic} that is computationally efficient, provided that one of the two sets is polytopic and has a small number of facets.
The asymmetric relaxation generalizes Lemma~\ref{lemma:valid_constraint}, with $\bc_i^\top \by + d_i \geq 0$ taking the place of the flow inequality~\eqref{eq:valid_inequality}.
The asymmetric relaxation is also set based, in the sense that it does not rely on the explicit constraints defining $\cX$, but it works directly with its abstract set representation.
Besides making the analysis very concise, this has also the practical advantages discussed in Remark~\ref{rem:set_based}.

\begin{remark}
\label{rem:coincise_bilinear}
The constraints of the biconvex program~\eqref{eq:bilinear_spp} can be restated in terms of the set $\cS$ as follows.
First, we collect in the vector $\by_v := (y_e)_{e \in \cE_v}$ the flows incident with vertex $v$.
Second, we let $\cY_v$ be the polytope defined by the linear constraints acting on $\by_v$.
Constraint~\eqref{eq:network_st} and the flow nonnegativity~\eqref{eq:network_flow_nonnegativity} make $\cY_s$ and $\cY_t$ unit simplices (recall that $|\cE_s\inc| = |\cE_t\out| = 0$).
For $v \neq s,t$, the polytope $\cY_v$ is defined by the flow nonnegativity~\eqref{eq:network_flow_nonnegativity} together with the conservation and degree constraints in~\eqref{eq:network_v}.
Third, we stack in the columns of the matrix $\bZ_v$ the auxiliary variables $\bz_e'$ for $e \in \cE_v\inc$ and $\bz_e$ for $e \in \cE_v\out$, so that the bilinear constraints~\eqref{eq:bilinear_yz} take the form $\bZ_v = \bx_v \by_v^\top$.
By defining the sets $\cS_v$ as in~\eqref{eq:bilinear_set}, the constraints of problem~\eqref{eq:bilinear_spp} become $(\bx_v, \by_v, \bZ_v) \in \cS_v$ for all $v \in \cV$.
Our relaxation of the SPP in GCS is then obtained by replacing the constraint sets $\cS_v$ with $\cS_v'$ defined as in~\eqref{eq:relaxation_polytopic}.
\end{remark}

\subsection{Tightness of the relaxation $\cS'$}
\label{sec:tightness_relaxation}

Ideally, we would like our relaxation to be as tight as possible, and the set $\cS'$ to coincide with the convex hull of $\cS$.
This equality holds, for example, when $\cX$ and $\cY$ are intervals on the real line, in which case $\cS'$ simplifies to the McCormick envelope~\cite{mccormick1976computability}.
However, the inclusion $\conv \cS \subset \cS'$ can be strict in general.
In fact, for polytopic sets $\cX$ and $\cY$, our approach of multiplying valid inequalities simplifies to the first level of the Reformulation-Linearization Technique (RLT)~\cite{sherali1990hierarchy}, which does not yield the convex hull of $\cS$ if, e.g., $\cX := \cY := [0,1]^2$.

The convex hull of $\cS$ can be efficiently described when $\cY$ is a polytope with a small number of extreme points $\{\hat \by_j\}_{j \in \cJ}$.
Specifically, by using disjunctive-programming techniques~\cite{ceria1999convex}, it can be verified that
\begin{align}
\label{eq:ch}
\conv \cS = \left\{
\sum_{j \in \cJ} (\bx_j, \lambda_j \hat \by_j, \bx_j \hat \by_j^\top) :
\sum_{j \in \cJ} \lambda_j = 1, \
(\bx_j, \lambda_j) \in \tilde \cX \tforall j \in \cJ
\right\}.
\end{align}
Note that this (lifted) description is convex and also set based.
While our relaxation $\cS'$ has size proportional to the number $|\cI|$ of facets of $\cY$, this description of the convex hull has size proportional to the number $|\cJ|$ of extreme points of $\cY$.
For the SPP in GCS, the polytopes $\cY_v$ have $O(|\cE_v\inc| + |\cE_v\out|)$ facets and only $O(|\cE_v\inc| |\cE_v\out|)$ extreme points, and this difference can be relatively small if the graph is sparse.
However, in our experience the MICPs obtained with our method provide a better tradeoff between strength and size, and are typically much faster to solve.

\begin{remark}
That our relaxation $\cS'$ is not always the convex hull of $\cS$ should be fully expected.
In fact, for $\cX := [0,1]^n$ and $\cY := [0,1]^m$, the bilinear program
\begin{align}
\label{eq:bilinear_optimization}
\minimize \quad \bp^\top \bx + \bq^\top \by + \bx^\top \bR \by
\quad \subjectto \quad \bx \in \cX, \ \by \in \cY,
\end{align}
is NP-hard~\cite{punnen2015bipartite}, and equivalent to minimizing a linear function over $\cS$.
The equality $ \cS' = \conv \cS$ would then allow us to solve an NP-hard problem in polynomial time.
\end{remark}

\subsection{Geometric proof of Theorem~\ref{th:validity_micp}}
\label{sec:conic}

In the proof of Theorem~\ref{th:validity_micp} we have shown the correctness of the MICP~\eqref{eq:micp} by analyzing all the feasible values that the variables in this program can take.
We now present a simple property of the relaxation $\cS'$ that will lead to a geometric and more concise proof of Theorem~\ref{th:validity_micp}.
This result will also generalize a known property of RLT.

\begin{lemma}
\label{lemma:extreme}
Let $\cY$ be a polytope and $\hat \by$ one of its extreme points.
We have $(\bx, \hat \by, \bZ) \in \cS$ if and only if $(\bx, \hat \by, \bZ) \in \cS'$.
\end{lemma}

\begin{proof}
One direction follows from $\cS \subseteq \cS'$.
For the other direction we show that if $(\bx, \hat \by, \bZ) \in \cS'$ then $\bZ = \bx \hat \by^\top$.
Since $\hat \by$ is an extreme point of $\cY$, there are $m$ linearly independent inequalities that are active at $\hat \by$.
Let $\bC \in \R^{m \times m}$ and $\bd \in \R^m$ collect the coefficients $(\bc_i, d_i)$ of these inequalities, so that $\bC \hat \by + \bd = \bzero$.
For the same inequalities, the constraints in~\eqref{eq:relaxation_polytopic} give us $\bZ \bc_i  + d_i \bx = \bzero$ or, equivalently, $\bZ \bC^\top  + \bx \bd^\top = \bzero$.
We then have $\bZ \bC^\top  = \bx \hat\by^\top \bC^\top$ and, since $\bC$ is invertible, $\bZ = \bx \hat \by^\top$.
\end{proof}

\begin{proof}[Alternative proof of Theorem~\ref{th:validity_micp}]
As in the first proof of Theorem~\ref{th:validity_micp}, note that the optimal solution of the MICP~\eqref{eq:micp} is such that the edges traversed by a unit of flow identify a path $p$.
Note also that, for all $v \in \cV$, the flow vectors $\by_v$ corresponding to a path $p$ are extreme points of the polytopes $\cY_v$.
Then the validity of the MICP~\eqref{eq:micp} follows since our relaxation is exact in these points by Lemma~\ref{lemma:extreme}.
\end{proof}

\begin{remark}
Consider the bilinear program~\eqref{eq:bilinear_optimization} with polytopic sets $\cX$ and $\cY$, and the additional constraint $\by \in \{0,1\}^m$.
Assuming $\cY \subseteq [0,1]^m$, the first-level RLT is known to yield a valid mixed-integer linear formulation of this program~\cite[Theorem~1]{adams1990linearization}.
Lemma~\ref{lemma:extreme} extends this result to generic closed convex sets $\cX$.
In fact, $\cY \subseteq [0,1]^m$ ensures that any vector $\by \in \cY \cap \{0,1\}^m$ is an extreme point of $\cY$, and the relaxation $\cS'$ is exact in correspondence of these points.
\end{remark}

\subsection{Related relaxation techniques}

The basic idea of generating new valid constraints by multiplying existing ones is classical, and has many incarnations: from the simple McCormick envelope~\cite{mccormick1976computability} to semidefinite hierarchies for polynomial optimization~\cite{lasserre2001global,parrilo2003semidefinite}, passing through RLT~\cite{sherali1990hierarchy}.
Among this family of techniques, the Lov\'asz-Schrijver hierarchy~\cite{lovasz1991cones} is the closest to ours, since it is set based and includes constraints of the form~\eqref{eq:relaxation_polytopic}; see~\cite[Theorem~1.6 and Conditions~(iii) to (iii'')]{lovasz1991cones}.
However, this hierarchy focuses on binary optimization and symmetric quadratic maps, and its naive application to the bilinear set $\cS$ would produce multiple redundant variables and constraints.
Our approach leverages the bilinear structure of the set $\cS$, that emerges naturally in the SPP in GCS, to construct a relaxation $\cS'$ that is smaller and as tight as the first level of the Lov\'asz-Schrijver hierarchy, without semidefinite constraints.
(As discussed in Section~\ref{sec:alternative_formulations}, our practical experience is that higher levels of the hierarchy and semidefinite constraints lead to MICPs that, although stronger, are significantly slower to solve.)

If both the sets $\cX$ and $\cY$ are polytopes, the convex hull of $\cS$ in~\eqref{eq:ch} is also a polytope, and its extreme points are $\ext \cS = \{(\bx, \by, \bx \by^\top): \bx \in \ext \cX, \by \in \ext \cY\}$.
In general, this yields an exponential-size description of $\conv \cS$.
Nevertheless, if the sets $\cX$ and $\cY$ have further special structure then specialized techniques can be applied to efficiently generate additional valid inequalities for $\conv \cS$; see, e.g., the techniques developed for network-interdiction problems~\cite{davarnia2017simultaneous}, pooling problems~\cite{gupte2017relaxations}, bipartite bilinear programs~\cite{dey2019new}, and bipartite boolean quadratic programs~\cite{sripratak2022bipartite}.

The recent work~\cite{zhen2021extension} shows how perspective functions can be used to allow the multiplication of nonlinear convex constraints in the RLT algorithm.
However, the relaxation in that work is not set based, and requires an explicit analysis of all the possible products of basic cone inequalities.

\section{Control applications}
\label{sec:optimal_control}
A main application of the framework presented in this paper is optimal control of discrete-time dynamical systems.
In this section we show how two simple control problems can be cast as SPPs in GCS.
These examples illustrate some basic modeling techniques that can also be applied to control problems involving more complex discrete decision making.

\subsection{Minimum-time control}

Consider the linear dynamical system $\bs_{\tau+1} = \bA \bs_\tau + \bB \ba_\tau$, where $\bs_\tau \in \R^q$ and $\ba_\tau \in \R^r$ are the system state and control action at time step $\tau$.
Given an initial state $\bs_0$, we look for a sequence of controls that drives the system state to the origin in the minimum number $T$ of time steps.
At each time $\tau$, the state and control pair $(\bs_\tau, \ba_\tau)$ is constrained in a compact convex set $\cD$.

To formulate this problem as an SPP in GCS we proceed as in Figure~\ref{fig:minimum_time}.
The vertices $\cV$ in our graph are ordered in a sequence.
The source $s$ is the first vertex and the target $t$ is the last.
The number of vertices is equal to $\bar T+1$, where $\bar T$ is a given upper bound on the optimal time horizon $T$.
Each vertex that is not the target has two outgoing edges: one that connects it to the next vertex in the sequence and one that goes to the target.
For each $v \in \cV$, the continuous variable $\bx_v$ represents a state and control pair $(\bs_v, \ba_v)$.
These variables are constrained by the following sets: $\cX_s :=\cD \cap (\{\bs_0\} \times \R^r)$ for the source, $\cX_t := \{(\bzero,\bzero)\}$ for the target (the value of $\ba_t$ is actually irrelevant), and $\cX_v := \cD$ for all the other vertices.
To minimize the number of edges in the optimal path (i.e., the time steps to reach the origin), the length of each edge $(u,v)$ is $1$ if $\bs_{v} = \bA \bs_u + \bB \ba_u$ and infinite otherwise.
(See Example~\ref{ex:perspective_extended_values} for the perspective of such a function.)

The solution of the MICP~\eqref{eq:micp} gives us a path $p := (v_0, \ldots, v_K)$.
The optimal time horizon is $T := K$, and the corresponding control sequence is $\ba_\tau := \ba_{v_\tau}$ for $\tau = 0, \ldots, T-1$.
The state trajectory is retrieved similarly, and is such that $\bs_T := \bs_t = \bzero$.

\begin{figure}[t]
\centering
\subfloat[Minimum-time problem.]{\label{fig:minimum_time}\includegraphics[height=3.5cm]{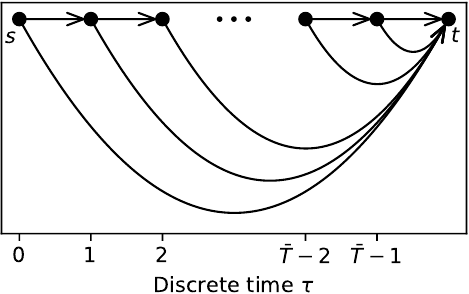}} \quad
\subfloat[Control of a PWA system.]{\label{fig:pwa}\includegraphics[height=3.5cm]{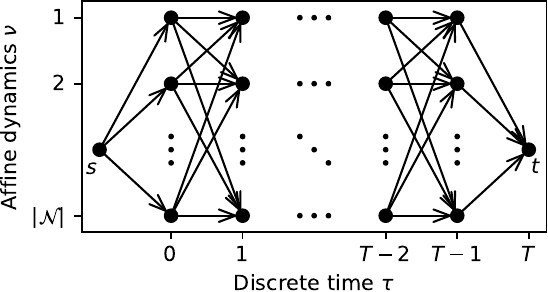}}
\caption{Graphs for the formulation of the optimal-control problems in Section~\ref{sec:optimal_control} as SPPs in GCS.}
\end{figure}

\subsection{Control of hybrid systems}
\label{sec:control_pwa}

PieceWise-Affine (PWA) systems are a popular  framework for modeling hybrid dynamics.
Loosely speaking, almost any dynamical system whose nonlinearity is exclusively due to discrete logics can be written in PWA form~\cite{heemels2001equivalence}.
Among the many applications of PWA systems, we have automotive~\cite{borrelli2006mpc}, power electronics~\cite{geyer2008hybrid}, and robotics~\cite{marcucci2017approximate}.
Given a finite collection $\{\cD_\nu\}_{\nu \in \cN}$ of compact convex subsets of the state and control space, a PWA system has dynamics $\bs_{\tau+1} = \bA_{\nu_\tau} \bs_\tau + \bB_{\nu_\tau} \ba_\tau + \bc_{\nu_\tau}$ if $(\bs_\tau, \ba_\tau) \in \cD_{\nu_\tau}$.
The index $\nu_\tau \in \cN$ represents the system discrete \emph{mode} at time $\tau$, which is itself a decision variable.
We consider the problem of driving a PWA system from a given initial state $\bs_0$ to the origin, in a fixed number $T$ of time steps.
The objective is to minimize the sum of the stage costs $\gamma(\bs_\tau, \ba_\tau)$ for $\tau=0, \ldots, T-1$.
The function $\gamma$ is convex and finite.

We model this problem through the GCS in Figure~\ref{fig:pwa}.
The source $s$ is the leftmost vertex and the target $t$ is the rightmost.
In between, we have $T$ layers with $|\cN|$ vertices each.
The source is connected via an edge to each vertex in the first layer, and all the vertices in the last layer are connected to the target.
Each pair of consecutive layers is fully connected.
Also in this case the continuous variables $\bx_v$ represent state and control pairs $(\bs_v, \ba_v)$.
The source is paired with the set $\cX_s := \{(\bs_0,\bzero)\}$, the target with $\cX_t := \{(\bzero,\bzero)\}$, and the $\nu$th vertex $v$ of each layer with $\cX_v := \cD_\nu$.
To enforce the initial conditions,  the edges $(s,v)$ outgoing from the source have zero length if $\bs_v = \bs_s$, and infinite length otherwise.
(Note that here the values of both $\ba_s$ and $\ba_t$ are irrelevant.)
The length of any other edge $(u,v)$, where $u$ is the $\nu$th vertex in its layer, is $\gamma(\bs_u, \ba_u)$ if $\bs_v = \bA_\nu \bs_u + \bB_\nu \ba_u + \bc_\nu$ and infinite otherwise.

A shortest path $p := (v_0, \ldots, v_K)$ has now $T+2$ vertices.
The optimal control at time $\tau = 0, \ldots, T-1$ is $\ba_\tau := \ba_{v_{\tau+1}}$.
The state trajectory is defined similarly.

\begin{remark}
Frequently in optimal control we need to enforce convex terminal constraints of the form $\bs_T \in \cD_T$, as well as convex terminal penalties $\gamma_T(\bs_T)$.
These are easily incorporated in our construction through a suitable modification of the set $\cX_t$ and the lengths of the edges incoming to the target vertex.
\end{remark}

\begin{remark}
The size of the GCS we just constructed is linear in the time horizon $T$ and quadratic in the number $|\cN|$ of discrete modes.
Conversely, common formulations for these problems have size linear in both $T$ and $|\cN|$~\cite{marcucci2019mixed}.
We will see in Section~\ref{sec:example_pwa} that the greater strength of our MICPs can be well worth this price.
\end{remark}

\section{Numerical results}
\label{sec:examples}
This section collects multiple numerical experiments.
We start in Section~\ref{sec:example_2d} with a simple two-dimensional problem.
Section~\ref{sec:example_large_scale} presents a statistical analysis of the performance of our MICP on large-scale instances of the SPP in GCS.
In Section~\ref{sec:example_pwa} we compare our approach with state-of-the-art mixed-integer formulations for control.
Finally, in Section~\ref{sec:example_failures} we use a carefully designed problem to show how symmetries in the GCS can loosen the relaxation of our MICP.

The code necessary to reproduce these results is available at \url{https://github.com/TobiaMarcucci/shortest-paths-in-graphs-of-convex-sets}.
All the experiments are run using the commercial solver MOSEK~10.0 with default options on a laptop computer with processor 2.4 GHz 8-Core Intel Core i9 and memory 64 GB 2667 MHz DDR4.
A mature implementation of the techniques presented in this paper is also provided by the open-source software Drake~\cite{tedrake2019drake}.

\subsection{Two-dimensional example}
\label{sec:example_2d}
We consider the two-dimensional problem in Figure~\ref{fig:2d_setup}.
We have a graph $G$ with $|\cV|=9$ vertices, $|\cE|=22$ edges, and multiple cycles.
The source $\cX_s := \{\btheta_s\}$ and target $\cX_t := \{\btheta_t\}$ sets are single points, while the remaining regions are full dimensional.
The geometry of the sets $\cX_v$ and the edge set $\cE$ can be deduced from Figure~\ref{fig:2d_setup}.
As edge lengths we consider the Euclidean distance~\eqref{eq:2norm} and the Euclidean distance squared~\eqref{eq:2norm_squared}, whose corresponding shortest paths are shown in Figure~\ref{fig:2d_setup} in orange and blue.
As expected, the first path is almost straight, while the lengths of the segments in the second are better balanced.

In Figure~\ref{fig:2d_results} we compare the optimal values of the SPP in GCS, the relaxation of our MICP~\eqref{eq:micp}, and the relaxation of the  McCormick formulation~\eqref{eq:mc_relaxation}.
Both relaxations are Second-Order-Cone Program (SOCPs), and for the McCormick one the bounding boxes $\cB_v$ are chosen as small as the corresponding sets $\cX_v$ allow.
We run this comparison for different values of a parameter $\sigma>0$ that controls the volume of the sets $\cX_v$.
The value $\sigma=1$ corresponds to the GCS in Figure~\ref{fig:2d_setup}.
While for $\sigma \neq 1$ each set $\cX_v$ is shrunk or enlarged via a uniform scaling, with scale factor $\sigma$, relative to a fixed Chebyshev center of the set.

\begin{figure}[t]
\centering
\subfloat[GCS with sets of nominal size, $\sigma =1$. The optimal solutions for the edge lengths~\eqref{eq:2norm} and~\eqref{eq:2norm_squared} are shown in orange and blue, respectively.]{\label{fig:2d_setup}\includegraphics[height=6cm]{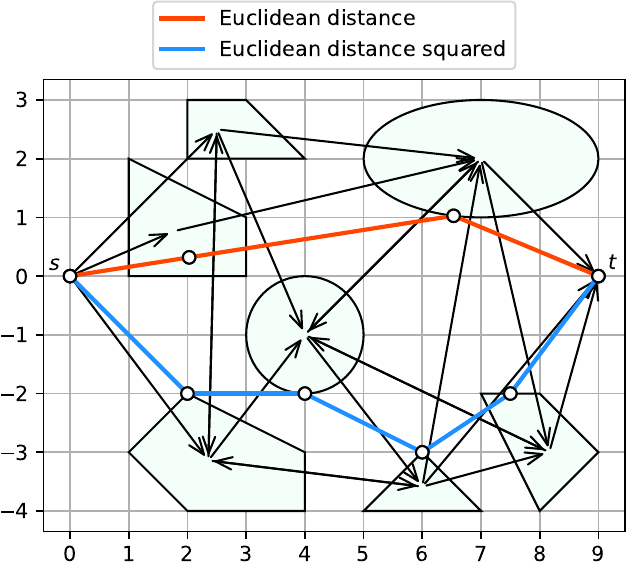}} \quad
\subfloat[Optimal values of the SPP in GCS and its convex relaxations as functions of the edge length and the size of the sets.]{\label{fig:2d_results}\includegraphics[height=6cm]{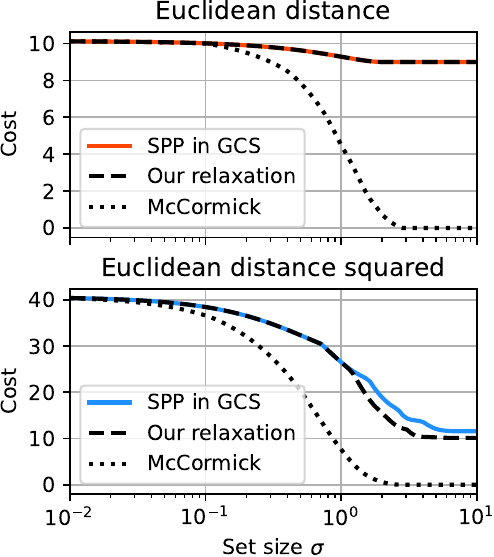}}
\caption{Two-dimensional SPP in GCS from Section~\ref{sec:example_2d}.
The tightness of the convex relaxation of our MICP~\eqref{eq:micp} is analyzed for two edge lengths (the Euclidean distance~\eqref{eq:2norm} and the Euclidean distance squared~\eqref{eq:2norm_squared}) and different sizes of the sets $\cX_v$ (parameterized by the scalar $\sigma$).
As a baseline, we also report the optimal value of the relaxation of the McCormick formulation~\eqref{eq:mc_relaxation}.
}
\end{figure}

When the edge length is the Euclidean distance~\eqref{eq:2norm}, the top panel in Figure~\ref{fig:2d_results} shows that our relaxation is exact for all values of $\sigma$.
This was expected for $\sigma$ close to zero, since by Remark~\ref{rem:singletons} our relaxation is exact when the sets are singletons.
Similarly, the problem is trivial for very large $\sigma$, when the regions are so big that, no matter the discrete path we take, we can always reach the target via a straight line.
However, that our relaxation is exact for all the intermediate values of $\sigma$ is not an obvious result.
The McCormick relaxation is also exact for small $\sigma$, but gives a trivial lower bound of zero when the sets are large.

With the Euclidean length squared~\eqref{eq:2norm_squared}, both relaxations are still guaranteed to be tight as $\sigma$ goes to zero.
This is confirmed by the bottom panel of Figure~\ref{fig:2d_results}.
When $\sigma$ is very large, we have seen in Section~\ref{sec:complexity} that our problem is equivalent to the HPP, and the argument from Theorem~\ref{th:complexity} shows that its optimal value is $\|\btheta_t - \btheta_s\|_2^2/K = 11.6$, where $K=7$ is the number of edges in the longest $s$-$t$ path in the graph in Figure~\ref{fig:2d_setup}.
A close inspection of the bottom of Figure~\ref{fig:2d_results} reveals that, for large $\sigma$, our relaxation yields the lower bound $\|\btheta_t - \btheta_s\|_2^2/(|\cV|-1) = 10.1$, which corresponds to the simple inequality $K \leq |\cV| -1$.
(Using a duality argument, it can be verified that our relaxation always recovers this bound.)
Conversely, the lower bound provided by the McCormick relaxation is again equal to zero.

\subsection{Large-scale random instances}
\label{sec:example_large_scale}
We present a statistical analysis of the performance of our formulation.
We generate a variety of random large-scale SPPs in GCS, and we analyze the relaxation tightness and the solution times of the MICP~\eqref{eq:micp} as functions of various problem parameters.
We stress that generating random graphs representative of the ``typical'' SPP in GCS we might encounter in practice is a difficult operation.
Inevitably, the instances we describe below are not completely representative, and our algorithm might perform worse or better on other classes of random graphs.
Our goal here is to show that our MICP is not limited to small-scale problems.

We construct an SPP in GCS as follows.
We set $\cX_s := \{\bzero\}$ and $\cX_t := \{\bone\}$.
The rest of the sets $\cX_v$ are axis-aligned cubes with volume $\Lambda$ and center drawn uniformly at random in $[0,1]^n$.
Given a number $|\cE|$ of edges, we construct the edge set in two steps.
First we generate multiple $s$-$t$ paths such that every vertex $v \neq s,t$ is traversed exactly by one path.
These are determined via a random partition of the set $\cV - \{s, t\}$: the number of sets in the partition (number of paths) is drawn uniformly from the interval $[1, |\cV| - 2]$, and also the number of vertices in each set (length of each path) is a uniform random variable.
Then we extend the edge set by drawing edges uniformly at random from the set $\{(u,v) \in \cV^2 : v \neq s, u \neq t, u \neq v \}$ until a desired cardinality $|\cE|$ is reached.
As edge lengths we consider the Euclidean distance~\eqref{eq:2norm} and the Euclidean distance squared~\eqref{eq:2norm_squared}, which both make our formulation~\eqref{eq:micp} a mixed-integer SOCP.

For each edge length, we first solve $100$ random instances with the following nominal parameters: volume $\Lambda = 0.01$, $n = 4$ dimensions, $|\cV| = 50$ vertices, and $|\cE|=100$ edges.
Then we solve four other batches of $100$ problems where, in each batch, a different subset of these parameters is increased by a factor of $5$.
Specifically, these additional batches test our formulation in case of large sets $\cX_v$ ($\Lambda$ from $0.01$ to $0.05$), high dimensions ($n$ from $4$ to $20$), dense graphs ($|\cE|$ from $100$ to $500$), and large graphs ($|\cV|$ and $|\cE|$ from $50$ and $100$ to $250$ and $500$).
To give an idea of what these problems look like, the projection onto two dimensions of a GCS generated using the nominal parameters is shown in Figure~\ref{fig:random_instance}.

\begin{figure}[t]
\centering
\includegraphics[height=3.2cm]{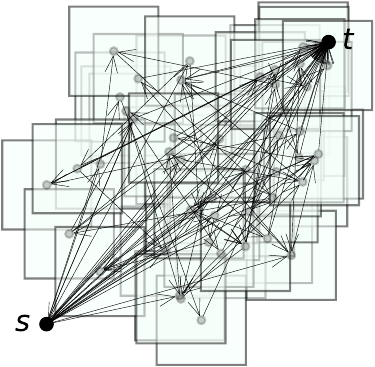}
\caption{Projection onto two dimensions of a random instance of the SPP in GCS from Section~\ref{sec:example_large_scale}. The problem parameters have nominal value.}
\label{fig:random_instance}
\end{figure}

\begin{figure}[t]
\centering
\includegraphics[width=.99\columnwidth]{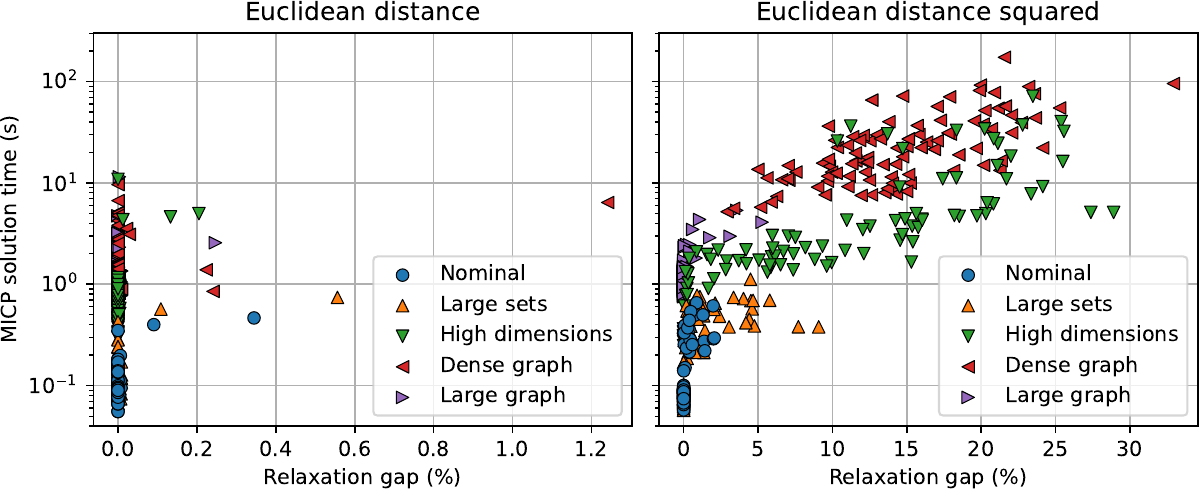}
\caption{Relaxation gap versus MICP solution time for the 500 random instances described in Section~\ref{sec:example_large_scale}.
Two edge lengths are analyzed: the Euclidean distance~\eqref{eq:2norm} and the Euclidean distance squared~\eqref{eq:2norm_squared}.
For each edge length, 100 nominal instances are generated with the nominal problem parameters, and four other batches of 100 instances each are obtained by increasing a different subset of the parameters.
Our relaxation is almost always exact with the Euclidean length.
While, with the Euclidean length squared, it is more sensitive to the dimension $n$ of the space and the density of the graph $G$.
(Note the different horizontal scales of the two plots.)
}
\label{fig:statistical_analysis}
\end{figure}

Figure~\ref{fig:statistical_analysis} shows the relaxation gap (cost gap between the MICP and its relaxation, normalized by the MICP cost) versus the MICP solution time for all the instances described above.
As observed in the previous example, the Euclidean edge length~\eqref{eq:2norm} results in easier programs: our relaxation is tight in almost all the instances and the solution times are relatively low.
The squared edge length~\eqref{eq:2norm_squared} leads to more challenging problems, even though the maximum relaxation gap and  runtime are only $2.1\%$ and $0.66$s in the nominal case.
When the volume of the cubes $\cX_v$ is increased to $\Lambda=0.05$ these values increase to $9.1\%$ and $1.12$s, and the performance of our MICP is minimally affected.
Note that this is not in contrast with the previous example, where we analyzed the regime of extremely large sets $\cX_v$.
Note also that the volume of the sets does not affect the MICP size.
The growth of the space dimension to $n=20$ increases the size of our programs, and also loosens the relaxation.
The largest relaxation gap is $28.9\%$, and our MICP takes $72$s to be solved in the worst case.
Similarly, when the number $|\cE|$ of edges is increased to $500$ the maximum relaxation gap and runtime become $32.9\%$ and $174$s.
This is due to the combination of the quadratic edge length and the large number of cycles that we have in a graph with high density of edges $|\cE|/|\cV|$.
To show this, in the last batch of problems we keep $|\cE|=500$ and we increase the number of vertices to $|\cV| = 250$.
This increases the MICP size further but makes the graph sparser, reducing the maximum relaxation gap and runtime to $5.3\%$ and $5.4$s.

Also for the problems in this analysis our formulation outperforms the McCormick one in~\eqref{eq:mc_relaxation}.
With the nominal parameters, the McCormick median (maximum) runtime is $12.9$ ($4.3$) times larger than ours for the Euclidean length~\eqref{eq:2norm}, and $10.3$ ($2.7$) times larger for the Euclidean length squared~\eqref{eq:2norm_squared}.
This performance difference grows larger for the other batches of problems, where the McCormick formulation reaches our time limit of one hour very often.
The slowness of the McCormick approach is due to its loose relaxation: even with the nominal parameters, we have a median (maximum) relaxation gap of $29\%$ ($52\%$) for~\eqref{eq:2norm}, and $34\%$ ($58\%$) for~\eqref{eq:2norm_squared}.

\subsection{Optimal control}
\label{sec:example_pwa}
We apply the method from Section~\ref{sec:control_pwa} to solve the optimal-control problem shown in Figure~\ref{fig:footstep_micp}.
We have a mechanical system with position $\bq \in \R^2$, velocity $\bv \in \R^2$, and force $\ba \in \R^2$.
The system has the dynamics of a double integrator: $\bq_{\tau + 1} = \bq_\tau + \bv_\tau$ and $\bv_{\tau + 1}  = \bv_\tau + \eta \ba_\tau$, where $\eta$ is a scalar parameter that regulates the system controllability.
The system state at time $\tau$ is $\bs_\tau := (\bq_\tau, \bv_\tau)$.
The initial  position is $\bq_0:= (0.5,-3.5)$ (green plus at the bottom left  of Figure~\ref{fig:footstep_micp}), the initial velocity is $\bv_0 := \bzero$.
At each time step $\tau=1, \ldots, T-1$, the position $\bq_\tau$ must belong to one of the seven regions in Figure~\ref{fig:footstep_micp}, while the velocity and the controls are limited by the constraints $\| \bv_\tau \|_\infty \leq 1$ and $\| \ba_\tau \|_\infty \leq 1$.
The goal is to reach the point $\bq_T := (6.5,3.5)$ (green cross at the top right of Figure~\ref{fig:footstep_micp}) with zero velocity $\bv_T$ in $T := 30$ time steps.
The cost function is the sum of the stage costs $\gamma(\bs_\tau, \ba_\tau) := \| \bv_\tau \|_2^2/5 + \| \ba_\tau \|_2^2$.

We let the parameter $\eta$ vary between the seven regions.
The five regions in the range $-5 \leq q_2 \leq 5$ (light blue in Figure~\ref{fig:footstep_micp}) have $\eta=1$.
While in the other two regions (red in Figure~\ref{fig:footstep_micp}) we make the system more expensive to control by setting $\eta=0.1$.
Since the parameter $\eta$ varies with the state, the system dynamics is PWA and the control problem falls into the class considered in Section~\ref{sec:control_pwa}.
The GCS beneath this problem (depicted in Figure~\ref{fig:pwa}) has $|\cV| = 212$ vertices and $|\cE| = 1435$ edges, and the convex sets $\cX_v$ live in $\R^6$.
Also in this case problem~\eqref{eq:micp} is a mixed-integer  SOCP.

\begin{figure}[t]
\centering
\subfloat[Optimal solution of the control problem.]{\label{fig:footstep_micp}\includegraphics[height=6cm]{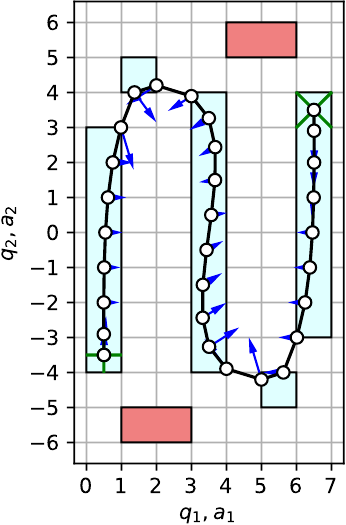}} \quad
\subfloat[Solution of the convex relaxation from~\cite{moehle2015perspective,marcucci2019mixed}. The relaxation gap is $93\%$, and the MICP is solved in $17$min.]{\label{fig:footstep_pf}\includegraphics[height=6cm]{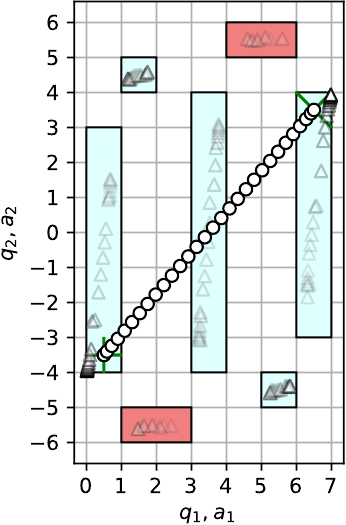}} \quad
\subfloat[Solution of our convex relaxation. The relaxation gap is $20\%$, and the MICP is solved in $7.1$s.]{\label{fig:footstep_spp}\includegraphics[height=6cm]{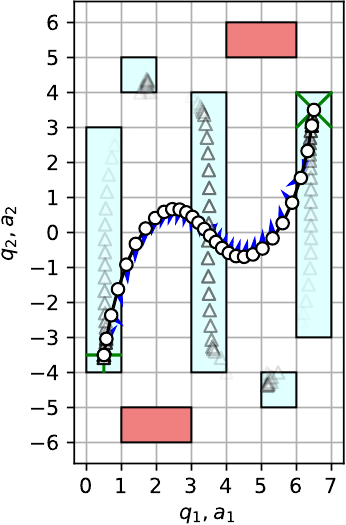}}
\caption{Control problem from Section~\ref{sec:example_pwa} of driving a dynamical system from start (green plus) to goal (green cross).
The light-blue and red regions have high and low controllability, respectively.
The optimal positions $\bq_\tau$ are white circles, the optimal controls $\ba_\tau$ are blue arrows.
The triangles are the auxiliary variables $\bq_\tau^\nu$ whose convex combination yields $\bq_\tau$.
The opacity of the triangles equals the optimal value of the variables $b_\tau^\nu$ that serve as weights in this convex combination.
}
\label{fig:footstep}
\end{figure}

Figure~\ref{fig:footstep_micp} shows the optimal trajectory $(\bq_0, \ldots, \bq_T)$ (white circles) and the optimal controls $(\ba_0, \ldots, \ba_{T-1})$ (blue arrows).
Geometrically, the red regions would be shortcuts to the goal, but the low controllability in these areas makes it too expensive not to fall out of the feasible set.
The optimal strategy is then to follow a winding trajectory and incur a cost of $9.37$.

As a baseline, we first solve the problem using the state-of-the-art perspective formulation from~\cite[Section~6]{moehle2015perspective} (see also~\cite[Section~5.2.2]{marcucci2019mixed}).
At each time step $\tau$, this expresses the system state $\bs_\tau$ as a convex combination of one auxiliary variable $\bs_\tau^\nu$ per region $\nu=1, \ldots, 7$.
The control $\ba_\tau$ is decomposed similarly.
When the coefficients $b_\tau^\nu$ of this combination are required to be binary, the solver is forced to make a hard selection of the region in which the system must be at each time step.
When the coefficients $b_\tau^\nu$ can be fractional, the system evolves according to a convex combination of the dynamics in each region.
Figure~\ref{fig:footstep_pf} illustrates the solution of the convex relaxation of this formulation (which, thanks to a perspective reformulation of the stage cost, is also an SOCP).
It reports the position $\bq_\tau$, the barely visible controls $\ba_\tau$, and the auxiliary copies $\bq_\tau^\nu$ of the position vector.
The latter have triangular markers and opacity equal to the value of the indicator $b_\tau^\nu$.
As it can be seen, this relaxation is insensitive to the arrangement of the regions, and its optimal trajectory heads straight to the goal.
Also the indicator variables $b_\tau^\nu$ are uninformative, and take nonzero value in the regions with low controllability (visible triangles in the red regions).
The optimal value of this relaxation is $0.67$, which is only $7\%$ of the MICP value ($93\%$ relaxation gap).
The MICP solution time is $1011\text{s}\approx 17$min.

The convex relaxation of our formulation is much tighter: its optimal value is $7.46$, which is $80\%$ of the MICP value ($20\%$ relaxation gap).
This has a dramatic effect on computation times that are now reduced to $7.1$s.
To make a plot comparable to Figure~\ref{fig:footstep_pf} we leverage the structure of our GCS in Figure~\ref{fig:pwa}.
The equivalent of the indicator variable $b_\tau^\nu$ is the total flow traversing the $\nu$th vertex in the $\tau$th layer of the graph.
Similarly, the position of the same vertex plays the role of the auxiliary variables $(\bs_\tau^\nu, \ba_\tau^\nu)$, which can then be combined using the coefficients $b_\tau^\nu$ to get candidate values for the state $\bs_\tau$ and the control $\ba_\tau$.
Figure~\ref{fig:footstep_spp} illustrates these values, and shows that the trajectory reconstructed from our relaxation resembles the MICP solution in Figure~\ref{fig:footstep_micp} much more closely.
All the markers in the regions with low controllability are now invisible, indicating that our relaxation correctly identifies these as regions of high cost.
The visible points $\bq_\tau^\nu$ are clustered along the optimal trajectory of the MICP, suggesting that our relaxation contains detailed information about the optimal path to reach the goal.

\subsection{Symmetries in the GCS}
\label{sec:example_failures}
\begin{figure}[t]
\centering
\subfloat[Optimal solution of the MICP, with the optimal vertex positions connected by orange lines.]{\quad\quad\label{fig:symmetry_mip}\includegraphics[height=4cm]{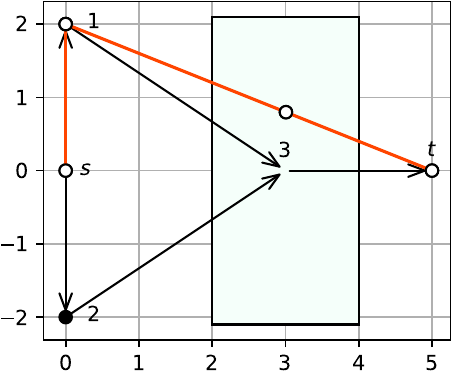}\quad\quad} \quad
\subfloat[Optimal solution of the relaxation.
For each edge $e = (u,v)$, the orange line connects the surrogates $\bar \bz_e$ and $\bar \bz_e'$ of the vertex positions $\bx_u$ and $\bx_v$, and is labeled with the flow $y_e$.]{\quad\quad\label{fig:symmetry_relaxation}\includegraphics[height=4cm]{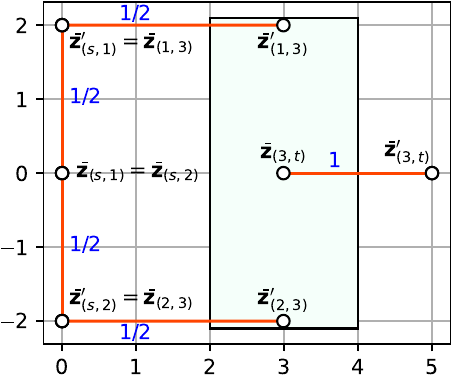}\quad\quad}
\caption{Instance of the SPP in GCS from Section~\ref{sec:example_failures} that shows how symmetries in the GCS can deteriorate the convex relaxation of our MICP.
For the relaxation, the cost contribution of edge $e$ is obtained by multiplying the flow $y_e$ by the distance between $\bar \bz_e$ and $\bar \bz_e'$.
Since only the mean of $\bar \bz_{(1,3)}'$ and $\bar \bz_{(2,3)}'$ is required to match $\bar \bz_{(3,t)}$, the cost is minimized by moving these two points closer to $\bar \bz_{(1,3)}$ and $\bar \bz_{(2,3)}$, respectively.
}
\label{fig:symmetry}
\end{figure}

We conclude by showing how symmetries in the GCS can deteriorate the convex relaxation of our MICP and, in principle, make it arbitrarily loose.
We illustrate this through the following carefully designed problem.

We consider the SPP in GCS depicted in Figure~\ref{fig:symmetry_mip}.
We have an acyclic graph with $|\cV| = 5$ vertices and $|\cE| = 5$ edges.
All the sets $\cX_v$ are singletons $\{ \btheta_v \}$, except for $\cX_3$ which is a full-dimensional rectangle.
As an edge length, we use the Euclidean distance~\eqref{eq:2norm}.
Solving this problem, we obtain the optimal path $p = (s,1,3,t)$ with length $7.4$ (the symmetric solution $p = (s,2,3,t)$ would also be optimal).
The corresponding vertex positions are connected by an orange line in Figure~\ref{fig:symmetry_mip}.

Figure~\ref{fig:symmetry_relaxation} illustrates the solution of the relaxation of the MICP~\eqref{eq:micp}.
For each edge $e$, we connect the optimal points $\bar \bz_e := \bz_e / y_e$ and $\bar \bz_e' := \bz_e' / y_e$ with an orange line, labeled in blue with the corresponding flow $y_e$.
Note that, for $y_e > 0$, we have $\tilde \ell_e (\bz_e, \bz_e', y_e) = \ell_e (\bar \bz_e, \bar \bz_e') y_e$, and the vectors $\bar \bz_e$ and $\bar \bz_e'$ are the actual points where the length of the edge $e$ is evaluated.
Note also that, by~\eqref{eq:micp_nonnegativity}, we have $\bar \bz_e \in \cX_u$ and $\bar \bz_e' \in \cX_v$.
The relaxation splits the unit of flow injected in the source into two: half unit is shipped to the target via the top path, the other half via the bottom path.
The optimal value of this convex program is $7.0$.

The looseness of the relaxation can be explained as follows.
If we denote with $\rho$ the flow traversing edge $(1,3)$, the flow conservation gives $y_{(2,3)} = 1 - \rho$, while the flow through the edge $(3,t)$ is always one.
Since the variables $\bar \bz_{(1,3)}$, $\bar \bz_{(2,3)}$, and $\bar \bz_{(3,t)}'$ are forced to match $\btheta_1$, $\btheta_2$, and $\btheta_t$, respectively, the cost terms in~\eqref{eq:micp_objective} corresponding to the edges $(1,3)$, $(2,3)$, and $(3,t)$ read
\begin{align}
\label{eq:symmetry_objective}
\rho \| \bar \bz_{(1,3)}' - \btheta_1 \|_2 + (1 - \rho) \| \bar \bz_{(2,3)}' - \btheta_2 \|_2 + \| \btheta_t - \bar \bz_{(3,t)} \|_2.
\end{align}
The only constraint that links these variables is~\eqref{eq:micp_conservation} for $v=3$, which gives $\rho \bar \bz_{(1,3)}' + (1 - \rho) \bar \bz_{(2,3)}' = \bar \bz_{(3,t)}$.
When $\rho = 1/2$, this constraint asks the mean of $\bar \bz_{(1,3)}'$ and $\bar \bz_{(2,3)}'$ to match $\bar \bz_{(3,t)}$, as opposed to forcing either one of the first two points to match the third, as it would be for $\rho \in \{0,1\}$.
Therefore, while keeping their mean equal to $\bar \bz_{(3,t)}$, the points $\bar \bz_{(1,3)}'$ and $\bar \bz_{(2,3)}'$ can move vertically, and get closer to $\btheta_1$ and $\btheta_2$.
This reduces the first two terms in~\eqref{eq:symmetry_objective}, and keeps the third term unchanged.

Although this example leads to a relaxation gap of only $5\%$, a simple variation of it shows that our relaxation can be arbitrarily loose.
In particular, if we let $\ell_{(s,1)} := \ell_{(s,2)} := 0$ and we shift the centers of the sets $\cX_3$ and $\cX_t$ to the origin, then the cost of the MICP and its relaxation are reduced to $2$ and $0$, and the relaxation gap becomes $100\%$.
Nevertheless, we emphasize that this is a contrived problem, and the instances we encounter in practice lead to these phenomena very rarely.

\section{Conclusions}
\label{sec:conclusions}
In this paper we have introduced the SPP in GCS, a versatile generalization of the classical SPP.
Our main contribution is a compact MICP formulation for the solution of this NP-hard problem.
Numerical experiments show that the convex relaxation of our formulation is typically very tight, and it enables us to quickly solve large problems to global optimality.
We have demonstrated  the applicability of the proposed framework to control systems: many optimal control problems are interpretable as SPPs in GCS and, in our tests, the proposed formulation outperforms state-of-the-art techniques for their solution.

\section*{Acknowledgments}

We would like to thank Hongkai Dai for all the time spent improving the solver interface used in the numerical experiments of this paper.

\bibliographystyle{siamplain}
\bibliography{references}

\appendix

\section{Other graph problems in GCS}
\label{sec:extensions}
Existing exact algorithms for graph problems with neighborhoods rely on expensive mixed-integer nonconvex optimization~\cite{gentilini2013travelling,burdick2021multi,blanco2017minimum}.
Here we show that, under standard convexity assumptions, the techniques from Section~\ref{sec:relaxation} apply beyond the SPP, and yield exact MICP reformulations for a wide variety of graph problems.
A thorough numerical evaluation of these novel formulations will be the object of future works.

Given a directed graph $G:=(\cV,\cE)$, many combinatorial problems require finding a set of edges $\cE^\star \subseteq \cE$ that is optimal according to a given criterion and given feasibility conditions.
Typically, these are formulated as integer linear programs of the form
\begin{align}
\label{eq:ilp}
\minimize \quad \sum_{e \in \cE} l_e y_e \quad \subjectto \quad \by \in \cY \cap \{0,1\}^{|\cE|},
\end{align}
where $\by := (y_e)_{e \in \cE}$.
The edge set $\cE^\star$ is parameterized by the variables $y_e$ as $\cE^\star = \{e \in \cE : y_e = 1\}$,
the polyhedron $\cY \subseteq [0,1]^{|\cE|}$ embodies the feasibility conditions, and
the cost is a linear function that assigns a weight $l_e \geq 0$ to each edge $e \in \cE$.

We extend the graph problem modeled by~\eqref{eq:ilp} to its version in GCS as done for the SPP.
We let the position $\bx_v \in \R^n$ of vertex $v$ be a decision variable, constrained in the set $\cX_v$, and we let the length of the edge $e = (u,v)$ be $\ell_e (\bx_u, \bx_v)$.
The sets $\cX_v$ and the functions $\ell_e$ satisfy the assumptions from Section~\ref{sec:statement}.
We define two auxiliary variables $\bz_e := y_e \bx_u$ and $\bz_e' := y_e \bx_v$ per edge $e = (u,v)$, and we formulate our graph problem in GCS as in~\eqref{eq:bilinear_spp}, with the condition $\by \in \cY \cap \{0,1\}^{|\cE|}$ in place of~\eqref{eq:bilinear_flow}.
This yields a mixed-integer program with bilinear constraints.
At this point, in the case of the SPP, we grouped the constraints in our problem vertex by vertex, and we applied the relaxation from Lemma~\ref{lemma:valid_constraint}.
However, in general, the polyhedron $\cY$ might not enjoy this convenient separability, as it might couple flows that do not share a common vertex.
There are two ways around this issue.

One option is just to separate the flow constraints that are vertex-wise separable from the ones that are not.
Using only the first to define the polyhedra $\cY_v \subseteq [0,1]^{|\cE_v|}$, we can then proceed as in Remark~\ref{rem:coincise_bilinear}.
The MICP we get is a valid problem formulation since any point in $\cY_v \cap \{0,1\}^{|\cE_v|}$ is an extreme point of $\cY_v$ and, by Lemma~\ref{lemma:extreme}, our relaxation is exact for those points.
The formulation resulting from this approach is compact but it might be weak.

The second option is to introduce new variables that represent the product of each flow $y_e$ and vertex position $\bx_v$, even if edge $e$ is not incident with vertex $v$.
This gives us a total of $n |\cV| |\cE|$ continuous variables $\bZ :=  \bx \by^\top$, where $\bx := (\bx_v)_{v \in \cV}$ lives in the Cartesian product $\cX := \prod_{v \in \cV} \cX_v$.
Defining the set $\cS$ as in~\eqref{eq:bilinear_set}, the constraints of our problem become $\by \in \{0,1\}^{|\cE|}$ and $(\bx, \by, \bZ) \in \cS$.
We then use the relaxation $\cS'$ of $\cS$ from~\eqref{eq:relaxation_polytopic} to get an MICP whose validity is ensured again by Lemma~\ref{lemma:extreme}.
This second option yields larger but potentially stronger MICPs.


\section{Dual optimization problem}
\label{sec:dual}
In this appendix we analyze the dual of the convex relaxation of the MICP~\eqref{eq:micp}, and we draw additional parallels between this problem and the network-flow formulation~\eqref{eq:network_flow} of the SPP.

\subsection{Dual of the SPP}
\label{sec:dual_lp}

As a reference for the discussion below, the dual of the LP~\eqref{eq:network_flow} is
\begin{subequations}
\label{eq:dual_lp}
\begin{align}
\maximize \quad & p_s - p_t \\
\label{eq:dual_lp_potential}
\subjectto \quad & p_u - p_v \leq l_e, && \forall e = (u,v) \in \cE.
\end{align}
\end{subequations}
Here $p_s$ and $p_t$ are the multipliers of the two constraints in~\eqref{eq:network_st}, and $p_v$ for $v \neq s,t$ are the multiplier of the flow conservation in~\eqref{eq:network_v}.
These multipliers are interpretable as potentials: the objective asks to maximize the potential jump between source and target, and the constraints ensure that the potential jump along each edge does not exceed the edge length.
Since the degree constraints in~\eqref{eq:network_v} are redundant (see Remark~\ref{rem:degree}), their multipliers do not appear in the dual problem.

For the LPs~\eqref{eq:network_flow} and~\eqref{eq:dual_lp}, complementary slackness reads $(l_e - p_u + p_v) y_e = 0$ for all edges $e =(u,v)$.
Therefore, at optimality, each edge $e \in \cE_p$ along the shortest path must have a potential jump equal to its edge length.

\subsection{Dual of the SPP in GCS}
\label{sec:dual_convex_relaxation}

The convex relaxation of the MICP~\eqref{eq:micp} is a conic program, and its dual is derived in the standard way.
To make the interpretation of the dual program easier, we assume that the graph $G$ is acyclic, and we remove the degree constraints~\eqref{eq:micp_degree} from the primal.
This leads to the following optimization problem:
\begin{subequations}
\label{eq:dual}
\begin{align}
\label{eq:dual_objective}
\maximize \quad & p_s - p_t \\
\label{eq:dual_potential}
\subjectto \quad
& \br_u^\top \bx_u + p_u - \br_v^\top \bx_v - p_v \leq \ell_e(\bx_u, \bx_v), \\
\nonumber & \qquad\qquad\qquad\qquad\qquad\qquad \ \ \forall \bx_u \in \cX_u, \bx_v \in \cX_v, e =(u,v) \in \cE, \\
\label{eq:dual_auxiliary}
& \br_s = \br_t = 0.
\end{align}
\end{subequations}
The dual variables are $p_v$ and $\br_v$ for all $v \in \cV$.
The first are paired with the flow constraints as above.
The second correspond to the portion of the flow conservation~\eqref{eq:micp_conservation} that involves the auxiliary variables $\bz_e$ and $\bz_e'$ (the additional variables $\br_s$ and $\br_t$ have only the role of simplifying the presentation).

Similarly to the LP~\eqref{eq:dual_lp}, the dual~\eqref{eq:dual} can be interpreted in terms of potentials.
For each vertex $v \in \cV$, the linear function $\br_v^\top \bx_v + p_v$ defines the potential of the point $\bx_v \in \cX_v$.
Because of~\eqref{eq:dual_auxiliary}, these functions are constant over the source and target sets, and the objective~\eqref{eq:dual_objective} maximizes the potential jump between $s$ and $t$ as in the classical SPP.
Like~\eqref{eq:dual_lp_potential}, constraint~\eqref{eq:dual_potential} asks the potential jump along an edge to be smaller than the edge length.
By setting all the potential functions to zero, we see that the dual problem is always feasible and has nonnegative optimal value.

For the primal-dual pair~\eqref{eq:micp} and~\eqref{eq:dual}, complementary slackness requires
$$
\br_u^\top \bz_e + p_u y_e - \br_v^\top \bz_e' - p_v y_e = \tilde\ell_e(\bz_e, \bz_e', y_e)
$$
for all edges $e=(u,v)$.
As for the classical SPP, this is trivially satisfied if $y_e = 0$.
While, for $y_e > 0$, we get $\br_u^\top \bar \bz_e + p_u - \br_v^\top \bar \bz_e' - p_v = \ell_e (\bar \bz_e, \bar \bz_e')$, with $\bar \bz_e := \bz_e/y_e$ and $\bar \bz_e' := \bz_e'/y_e$.
In words, at optimality, the potential jump along edge $e$ is tight to the edge length $\ell_e$ at the point $(\bar \bz_e, \bar \bz_e')$.

\end{document}